\documentclass[11pt]{article}

\usepackage{geometry,fullpage}
\usepackage{titling}
\usepackage{bm}
\usepackage{amsthm,amssymb,amsmath,amsfonts,rotate}
\usepackage{graphicx,algorithmic,algorithm}
\usepackage[]{datetime}
\usepackage{aliascnt}

\everymath{\color{MidnightBlack}}

\usepackage[pdftex,
backref=true,
plainpages = true,
pdfpagelabels,
hyperfootnotes=true,
pdfpagemode=FullScreen,
bookmarks=true,
bookmarksopen = true,
bookmarksnumbered = true,
breaklinks = true,
hyperfigures,
pagebackref,
urlcolor = magenta,
urlcolor = MidnightBlack,
anchorcolor = green,
hyperindex = true,
colorlinks = true,
linkcolor = black!30!blue,
citecolor = black!30!green
]{hyperref}

\RequirePackage[T1]{fontenc}
\RequirePackage[utf8]{inputenc}
\RequirePackage{amssymb}
\usepackage{alphabeta}
\usepackage{textcomp}
\DeclareUnicodeCharacter{2192}{\ifmmode\to\else\textrightarrow\fi}
\DeclareUnicodeCharacter{2203}{\ensuremath\exists}
\DeclareUnicodeCharacter{183}{\cdot}
\DeclareUnicodeCharacter{2200}{\forall}
\DeclareUnicodeCharacter{2264}{\leq}
\DeclareUnicodeCharacter{2265}{\geq}
\DeclareUnicodeCharacter{8614}{\mathbin{\mapsto}}
\DeclareUnicodeCharacter{8656}{\Leftarrow}
\DeclareUnicodeCharacter{8657}{\Uparrow}
\DeclareUnicodeCharacter{8658}{\Rightarrow}
\DeclareUnicodeCharacter{8659}{\Downarrow}
\DeclareUnicodeCharacter{8669}{\rightsquigarrow}
\newcommand{\eqdef}{\stackrel{{\scriptsize\rm def}}{=}}
\DeclareUnicodeCharacter{8797}{\eqdef}
\DeclareUnicodeCharacter{8870}{\vdash}
\DeclareUnicodeCharacter{8873}{\Vdash}
\DeclareUnicodeCharacter{8871}{\models}
\DeclareUnicodeCharacter{9121}{\lceil}
\DeclareUnicodeCharacter{9123}{\lfloor}
\DeclareUnicodeCharacter{9124}{\rceil}
\DeclareUnicodeCharacter{9126}{\rfloor}
\DeclareUnicodeCharacter{9655}{\triangleright}
\DeclareUnicodeCharacter{9665}{\triangleleft}
\DeclareUnicodeCharacter{9671}{\diamond}
\DeclareUnicodeCharacter{9675}{\circ}
\DeclareUnicodeCharacter{10178}{\bot}
\DeclareUnicodeCharacter{10214}{\llbracket} 
\DeclareUnicodeCharacter{10215}{\rrbracket} 
\DeclareUnicodeCharacter{10229}{\longleftarrow}
\DeclareUnicodeCharacter{10230}{\longrightarrow}
\DeclareUnicodeCharacter{10231}{\longleftrightarrow}
\DeclareUnicodeCharacter{10232}{\Longleftarrow}
\DeclareUnicodeCharacter{10233}{\Longrightarrow}
\DeclareUnicodeCharacter{10234}{\Longleftrightarrow}
\DeclareUnicodeCharacter{10236}{\longmapsto}
\DeclareUnicodeCharacter{10238}{\Longmapsto} 
\DeclareUnicodeCharacter{10503}{\Mapsto}    
\DeclareUnicodeCharacter{10971}{\mathrel{\not\hspace{-0.2em}\cap}}
\DeclareUnicodeCharacter{65294}{\ldotp}
\DeclareUnicodeCharacter{65372}{\mid}

\usepackage{lmodern}
\usepackage{amsmath,ifthen,eurosym}
\usepackage{wrapfig,cite}
\usepackage{latexsym,amssymb,url}
\usepackage{subcaption}
\usepackage{cleveref}

\usepackage[svgnames]{xcolor}

\usepackage{color}

\definecolor{MidnightBlack}{rgb}{0.1,0.1,0.30}
\definecolor{MidnightBlue}{rgb}{0.1,0.1,0.44}
\definecolor{Black}{rgb}{0,0, 0}
\definecolor{Blue}{rgb}{0, 0 ,1}
\definecolor{Red}{rgb}{1, 0 ,0}
\definecolor{White}{rgb}{1, 1, 1}
\definecolor{Grey}{rgb}{.6, .6, .6}
\definecolor{Mygreen}{rgb}{.0, .7, .0}
\definecolor{Yellow}{rgb}{.55,.55,0}
\definecolor{Mustard}{rgb}{1.0, 0.86, 0.35}
\definecolor{applegreen}{rgb}{0.55, 0.71, 0.0}
\definecolor{darkturquoise}{rgb}{0.0, 0.81, 0.82}
\definecolor{celestialblue}{rgb}{0.29, 0.59, 0.82}
\definecolor{green_yellow}{rgb}{0.68, 1.0, 0.18}
\definecolor{crimsonglory}{rgb}{0.75, 0.0, 0.2}
\definecolor{darkmagenta}{rgb}{0.30, 0.0, 0.30}
\definecolor{internationalorange}{rgb}{1.0, 0.31, 0.0}
\definecolor{darkorange}{rgb}{1.0, 0.55, 0.0}

\newcommand{\mnb}[1]{{\color{MidnightBlue}#1}}

\usepackage{tikz}
\usetikzlibrary{calc,3d}
\usetikzlibrary{intersections,decorations.pathmorphing,shapes,decorations.pathreplacing,fit}

\usepackage{todonotes}

\usepackage{amsmath,latexsym,amsthm,amsfonts,amssymb,%
	 ifthen,color,wrapfig,hyperref,rotate,lmodern,aliascnt,%
	 datetime,graphicx,algorithmic,algorithm,enumerate,enumitem,%
	todonotes,cleveref,bm,subcaption,tabularx,colortbl,xspace%
}

\newcommand{\bd}{{\sf bd}}
\newcommand{\yes}{{\sf yes}}
\newcommand{\no}{{\sf no}}

\newcommand{\remove}[1]{}

\newcommand{\bigmid}{\;\big|\;}
\newcommand{\cupall}{\pmb{\pmb{\bigcup}}}

\newcommand{\pretp}{\preceq_{\sf tm}}

\newcounter{func}
\newcommand{\newfun}[1]{f_{\refstepcounter{func}\label{#1}\thefunc}}
\newcommand{\funref}[1]{\hyperref[#1]{f_{\ref*{#1}}}} 

\newcounter{con}

\newcommand{\conref}[1]{\hyperref[#1]{c_{\ref*{#1}}}} 

\newcommand{\bound}[1]{\textbf{#1}}

\newcommand{\tw}{{\sf tw}}

\usepackage{dsfont}

\usepackage{float}
\usepackage{tikz,pgf}
\usetikzlibrary{backgrounds}
\usetikzlibrary{calc,3d}

\tikzset{red node/.style={draw=red, circle, fill = red, minimum size = 4pt, inner sep = 0pt}}
\tikzset{yellow node/.style={draw=yellow, circle, fill = yellow, minimum size = 4pt, inner sep = 0pt}}
\tikzset{blue node/.style={draw=celestialblue, circle, fill =celestialblue, minimum size = 4pt, inner sep = 0pt}}
\tikzset{triangle/.style = { regular polygon, regular polygon sides=3, rotate=180}}
\tikzset{small red/.style={draw=red, triangle, fill = red, minimum size = 2pt, inner sep = 0pt}}

\tikzset{black node/.style={draw, circle, fill = black, minimum size = 3pt, inner sep = 0pt}}
\tikzset{small black node/.style={draw, circle, fill = black, minimum size = 3pt, inner sep = 0pt}}
\tikzset{model node/.style={draw=celestialblue, circle, fill = celestialblue, minimum size = 5pt, inner sep = 0pt}}
\tikzset{model node small/.style={draw=celestialblue, circle, fill = celestialblue, minimum size = 3pt, inner sep = 0pt}}
\tikzset{rep node/.style={draw=red, circle, fill = red, minimum size = 3pt, inner sep = 0pt}}
\tikzset{track node 1/.style={draw, circle, fill = black, minimum size = 2pt, inner sep = 0pt}}
\tikzset{track node 2/.style={draw=black!30!white, circle, fill = black!30!white, minimum size = 2pt, inner sep = 0pt}}
\tikzset{track node 3/.style={draw=black!10!white, circle, fill = black!10!white, minimum size = 2pt, inner sep = 0pt}}

\newcommand{\mynewtheorem}[2]{
	\newaliascnt{#1}{dummy}
	\newtheorem{#1}[#1]{#2}
	\aliascntresetthe{#1}
}

\theoremstyle{plain}
\mynewtheorem{theorem}{Theorem}
\mynewtheorem{proposition}{Proposition}
\mynewtheorem{corollary}{Corollary}
\mynewtheorem{lemma}{Lemma}

\theoremstyle{definition}
\mynewtheorem{definition}{Definition}
\theoremstyle{remark}
\mynewtheorem{sublemma}{Sublemma}
\mynewtheorem{claim}{Claim}
\mynewtheorem{remark}{Remark}
\mynewtheorem{fact}{Fact}
\mynewtheorem{conjecture}{Conjecture}
\mynewtheorem{question}{Question}
\mynewtheorem{observation}{Observation}
\mynewtheorem{problem}{Problem}
\mynewtheorem{note}{Note}

\newcommand{\bigO}[1]{\mathcal O(#1)}
\newcommand{\NP}{{\sf NP}\xspace}
\newcommand{\FPT}{{\sf FPT}\xspace}

\newcommand{\W}{{\sf W}\xspace}
\newcommand{\frR}{{\frak{R}}}

\usepackage{titling}
\thanksmarkseries{arabic}
\newcommand*\samethanks[1][\value{footnote}]{\footnotemark[#1]}

\newtheorem{innercustomgeneric}{\customgenericname}
\providecommand{\customgenericname}{}

\begin{document}

\title{$k$-apices of minor-closed graph classes.\! II. Parameterized algorithms\thanks{A conference version of this paper appeared in the \emph{Proceedings of the 47th International Colloquium on Automata, Languages and Programming (\textbf{ICALP}), volume 168 of LIPICs, pages 95:1--95:20, \textbf{2020}}.}}

\author{\bigskip 	Ignasi Sau\thanks{LIRMM, Universit\'e de Montpellier, CNRS, Montpellier, France. Supported  by   the ANR projects DEMOGRAPH (ANR-16-CE40-0028), ESIGMA (ANR-17-CE23-0010), ELIT (ANR-20-CE48-0008), the French-German Collaboration ANR/DFG Project UTMA (ANR-20-CE92-0027), and the French Ministry of Europe and Foreign Affairs, via the Franco-Norwegian project PHC AURORA.  Emails:  \texttt{ignasi.sau@lirmm.fr}, \texttt{sedthilk@thilikos.info}.}\and
	Giannos Stamoulis\thanks{LIRMM, Universit\'e de Montpellier, Montpellier, France. Email: \texttt{gstamoulis@lirmm.fr}.}
	\and
	Dimitrios  M. Thilikos\samethanks[2]
}
\date{}

\maketitle

\begin{abstract}
	\noindent Let ${\cal G}$ be a minor-closed graph class. We say that a graph $G$ is a {\em $k$-apex} of ${\cal G}$ if  $G$ contains a  set $S$ of at most $k$ vertices such that  $G\setminus S$ belongs to ${\cal G}$. We denote by ${\cal A}_k ({\cal G})$ the set of all graphs that are $k$-apices of ${\cal G}.$ In the first paper of this series we obtained upper bounds on the size of the graphs in the minor-obstruction set of ${\cal A}_k ({\cal G})$, i.e., the minor-minimal set of graphs not belonging to ${\cal A}_k ({\cal G}).$
	In this article we provide an algorithm that,  given a graph $G$  on $n$ vertices, runs in time $2^{{\sf poly}(k)}\cdot n^3$ and either returns a set $S$ certifying that $G \in {\cal A}_k ({\cal G})$,  or reports that $G \notin {\cal A}_k ({\cal G})$. Here ${\sf poly}$ is a polynomial function whose degree depends on the maximum size of a minor-obstruction of ${\cal G}.$ In the special case where ${\cal G}$ excludes some apex graph as a minor,  we give an alternative  algorithm running in  $2^{{\sf poly}(k)}\cdot n^2$-time.

	\bigskip


\end{abstract}

\newpage

\tableofcontents

\newpage

\section{Introduction}
\label{label_deliberations}

Graph modification problems are fundamental in algorithmic graph theory. Typically,  such a problem is determined by a graph class ${\cal G}$ and some prespecified set ${\cal M}$ of {\sl  local}  modifications,
such as vertex/edge removal or edge addition/contraction or combinations of them,
and the question  is, given
a graph $G$ and an integer $k,$ whether it is possible to
transform $G$ to a graph in ${\cal G}$ by applying $k$ modification operations from ${\cal M}.$
A plethora of graph problems can be formulated
for different instantiations of ${\cal G}$ and ${\cal M}.$
Applications span diverse topics such as computational biology, computer vision, machine learning, networking, and sociology~\cite{FominSM15grap}.
As reported
by Roded Sharan in~\cite{Sharan02grap}, already in 1979 Garey and Johnson mentioned 18 different types of modification problems
\cite[Section A1.2]{GareyJ79comp}. For more on graph modification problems,   see~\cite{FominSM15grap,BodlaenderHL14grap} as well as the running survey in~\cite{CrespelleDFG13asur}. In this paper we focus our attention on the vertex deletion operation. We say that a graph $G$ is a {\em $k$-apex} of a graph class ${\cal G}$ if there is a set $S\subseteq V(G)$ of size at most $k$ such that the removal of $S$ from $G$ results in a graph in ${\cal G}.$
In other words, we consider the following meta-problem.

\begin{center}
	\fbox{
		\begin{minipage}{13.2cm}
			\noindent\mnb{\sc  Vertex Deletion to ${\cal G}$}\\
			\noindent\textbf{Input:}~~A graph $G$ and  a non-negative integer $k.$\\
			\textbf{Objective:}~~Find, if it exists, a set $S\subseteq V(G),$  certifying that  $G$ is $k$-apex of  ${\cal G}.$
		\end{minipage}
	}
\end{center}

\noindent To illustrate  the  expressive power of \mnb{\sc Vertex Deletion to ${\cal G}$}, if ${\cal G}$ is the class of  edgeless (resp. acyclic, planar,  bipartite, (proper) interval, chordal) graphs, we obtain the \mnb{\textsc{Vertex Cover}} (resp. \mnb{\textsc{Feedback Vertex Set}}, \mnb{\textsc{Vertex Planarization}}, \mnb{\textsc{Odd Cycle Transversal}}, \mnb{\sc (proper) Interval Vertex Deletion}, \mnb{\sc Chordal Vertex Deletion}) problem.

By the classical result of Lewis and Yannakakis~\cite{LewisY80then}, \mnb{\sc  Vertex Deletion to ${\cal G}$} is \NP-hard for every non-trivial graph class ${\cal G}.$ To circumvent its intractability, we study it from the parameterized complexity point of view and we parameterize
it by the number $k$ of  vertex deletions.
%
In this setting, the most desirable behavior is the existence of an algorithm running in time $f(k) \cdot n^{\bigO{1}},$ where $f$ is a computable function depending only on $k.$ Such an algorithm is called \emph{fixed-parameter tractable}, or \FPT-algorithm for short, and a parameterized problem admitting an \FPT-algorithm is said to belong to the parameterized complexity class  \FPT. Also, the function $f$ is called {\em parametric dependence} of the corresponding \FPT-algorithm, and the challenge is to design \FPT-algorithms
with small parametric dependencies~\cite{CyganFKLMPPS15para,DowneyF13fund,FlumG06para,Niedermeier06invi}.

Unfortunately, we cannot hope for the existence of \FPT-algorithms for every graph class ${\cal G}.$ Indeed, the problem is \W-hard\footnote{Implying that an \FPT-algorithm would result in an unexpected complexity collapse; see~\cite{DowneyF13fund}.} for some classes ${\cal G}$ that are closed under induced subgraphs~\cite{Lokshtanov08whee} or, even worse, \NP-hard, for $k=0,$ for every class ${\cal G}$ whose recognition problem is \NP-hard, such as
some classes closed under subgraphs or induced subgraphs (for instance 3-colorable graphs), edge contractions~\cite{BrouwerV87cont}, or induced minors~\cite{FellowsKMP95thec}.

On the positive side, a very relevant subset of classes of graphs {\sl does} allow for \FPT-algorithms. These are classes ${\cal G}$ that are closed under minors\footnote{A graph $H$ is a \emph{minor} of a graph $G$ if it can be obtained from a subgraph of $G$ by contracting edges, see~\autoref{label_incurability} for the formal definitions.}, or \emph{minor-closed}.  To see this, we define ${\cal A}_{k}({\cal G})$
as the class of the $k$-apices of ${\cal G},$ i.e., the \yes-instances of  \mnb{\sc  Vertex Deletion to ${\cal G}$},  and observe that if ${\cal G}$ is minor-closed
then the same holds for ${\cal A}_{k}({\cal G})$ for every $k\geq 0.$ This, in turn, implies that for every $k,$
${\cal A}_{k}({\cal G})$ can be characterized by a set ${\cal F}_{k}$ of minor-minimal graphs that are not in ${\cal A}_{k}({\cal G})$;
we call these graphs the {\em obstructions} of ${\cal A}_{k}({\cal G})$ and we know that they are finite because of the Robertson and Seymour's theorem~\cite{RobertsonS04XX}. In other words,
we know that the size of the obstruction set of ${\cal A}_{k}({\cal G})$ is bounded by some function of $k.$
Then one can decide whether a graph $G$ belongs to ${\cal A}_{k}({\cal G})$
by checking whether $G$ excludes all members of the obstruction set of ${\cal A}_{k}({\cal G}),$ and this can be checked by using the \FPT-algorithm in~\cite{RobertsonS95XIII} (see also~\cite{FellowsL88nonc}).

As the Robertson and Seymour's theorem~\cite{RobertsonS04XX} does {\sl not} construct ${\cal F}_{k},$ the aforementioned argument is not constructive, i.e., it is not able to construct the claimed \FPT-algorithm. An important step towards the constructibility of
such an  \FPT-algorithm was done by Adler et al.~\cite{AdlerGK08comp}, who proved that
${\cal F}_{k}$ is effectively computable.
In the first paper of this series~\cite{SauST21kapiI} we give an {\sl explicit upper bound} on the size of the graphs in ${\cal F}_{k},$ namely we prove that every graph in ${\cal F}_{k}$  has size bounded by an exponential tower of height four of a polynomial function in $k,$  whose degree depends on the size of the minor-obstructions of ${\cal G}.$ 
The focus of the current paper is on the parametric dependence of \FPT-algorithms to solve the \mnb{\sc  Vertex Deletion to ${\cal G}$} problem, i.e., for recognizing the class ${\cal A}_{k}({\cal G})$.


The task of specifying (or even optimizing) this parametric dependence
for different instantiations of ${\cal G}$ occupied a considerable part of research in parameterized algorithms.
The most general  result in this direction
states that, for every $t,$  there is some contant $c$ such that if the graphs in ${\cal G}$  have treewidth at most $t,$ then \mnb{\sc  Vertex Deletion to ${\cal G}$}
admits an \FPT-algorithm that runs in time
$c^{k}\cdot n^{{\cal O}(1)}$~\cite{FominLMS12plan,KimLPRRSS16line}.
Reducing the constant $c$ in this running time has  attracted  research on particular problems
such as
\mnb{\sc Vertex Cover}\cite{ChenKX10impr} (with $c=1.2738$),
\mnb{\sc Feedback Vertex Set}~\cite{KociumakaP14fast} (with $c=3.619$), \mnb{\sc Apex-Pseudoforest}~\cite{BodlaenderOO18afas} (with $c=3$), \mnb{\sc Pathwidth 1 Vertex Deletion} (with $c=4.65
$)\cite{CyganPPW12anim}, or \mnb{\sc Pumpkin Vertex Deletion}~\cite{JoretPSST14hitt}.
The first step towards
a parameterized algorithm for  \mnb{\sc  Vertex Deletion to ${\cal G}$} for cases where  ${\cal G} $  has  unbounded treewidth was  done  in~\cite{MarxS07obta} and later in~\cite{Kawarabayashi09plan} for the \mnb{\sc Vertex Planarization} problem, and the best parameterized dependence for this problem is $2^{{\cal O}(k\cdot \log k )} \cdot n,$ achieved by Jansen et al.~\cite{JansenLS14anea}.
These results were later  extended by Kociumaka and Marcin Pilipczuk~\cite{KociumakaP19dele}, who proved  that
if ${\cal G}_{g}$ is the class of graphs of Euler genus at most $g,$ then
\mnb{\sc Vertex Deletion to ${\cal G}_{g}$} admits a $2^{{\cal O}_{g}(k^2\cdot \log k)}\cdot n^{{\cal O}(1)}$-time\footnote{Given a tuple $\textbf{t}=(x_{1},\ldots,x_{\ell})\in \Bbb{N}^{\ell}$ and two functions $\chi,\psi: \Bbb{N}
	\rightarrow \Bbb{N},$
we write  $\chi(n)={\cal O}_{\textbf{t}}(\psi(n))$ in order to denote that there exists a computable
function $\phi:\Bbb{N}^{\ell} \rightarrow \Bbb{N}$
such that  $\chi(n)={\cal O}( \phi(\textbf{t})\cdot \psi(n)).$} algorithm.


\paragraph{Our results.} In this paper we give an explicit
\FPT-algorithm for
\mnb{\sc Vertex Deletion to ${\cal G}$} for {\sl every} fixed
minor-closed graph class ${\cal G}.$ In particular, our main results are the following.

\begin{theorem}\label{label_preliminaries}
	If  ${\cal G}$ is a minor-closed graph class, then \mnb{\sc  Vertex Deletion to ${\cal G}$} admits an algorithm running in time $2^{{\sf poly}(k)} \cdot n^3,$ for some polynomial {\sf poly} whose degree depends  on ${\cal G}.$
\end{theorem}

We say that a graph $H$ is an {\em apex} graph if it is a 1-apex of the class of planar graphs.

\begin{theorem}\label{label_entrepreneur}
	If  ${\cal G}$ is a minor-closed graph class excluding some  apex graph, then \mnb{\sc  Vertex Deletion to ${\cal G}$} admits an algorithm running in time $2^{{\sf poly}(k)} \cdot n^2,$ for some polynomial {\sf poly} whose degree depends  on ${\cal G}.$
\end{theorem}

In \autoref{label_ressentiment} we explain how the algorithms of \autoref{label_preliminaries} and  \autoref{label_entrepreneur}  can be modified in order to apply to a series of variants of \mnb{\sc  Vertex Deletion to ${\cal G}$}.
\smallskip

\paragraph{Our techniques.} We provide here just a very succinct enumeration of the techniques that we use in order to achieve \autoref{label_preliminaries} and \autoref{label_entrepreneur}; a more detailed description with the corresponding definitions is provided, along with the algorithms, in the next sections.

Our starting point to prove \autoref{label_preliminaries} is to use the standard iterative compression technique of Reed et al.~\cite{ReedSV04find} (\autoref{label_countenances}).
This allows us to assume that we have at hand a slightly too large set $S \subseteq V(G)$ such that $G \setminus S \in {\cal G}.$
We run the algorithm of \autoref{label_dishonorable} from~\cite{SauST21amor} that (since $G\setminus S\in {\cal G}$) either concludes that the treewidth of $G$ is polynomially bounded by $k,$ or finds a large {\sl flat} wall $W$ together with an apex set $A.$
In the first case, we use the main algorithmic result of Baste et al.~\cite{BasteST20acom} (\autoref{label_confiscating}) to solve the problem parameterized by treewidth, achieving the claimed running time.
\autoref{label_dishonorable} is an improved version of the original ``Flat Wall Theorem'' of Robertson and Seymour~\cite{RobertsonS95XIII}, whose proof is based on the recent results of Kawarabayashi et al.~\cite{KawarabayashiTW18anew}, which we state using the framework that we recently introduced in~\cite{SauST21amor}. This framework is presented in \autoref{label_forestalling} and provides the formal definitions of a series of
combinatorial concepts such as paintings and renditions (\autoref{label_klammerausdrucks}), flatness pairs and  tilts (\autoref{label_souhaiteroient}), as well
as a notion of wall homogeneity (\autoref{label_definitionen}) alternative to the one given in~\cite{RobertsonS95XIII}. All these concepts
are extensively  used in our proofs, as well as in those in the first article of this series~\cite{SauST21kapiI}.

%
%

Once we have the large  flat wall $W$ and the apex set $A$,
we see how many vertices of $S\cup A$ have enough neighbors in the ``interior'' of $W$.
Two possible scenarios may occur. If the ``interior'' of $W$ has enough  neighbors in the set  $S \cup A,$ we apply a  combinatorial result of~\cite{SauST21kapiI} (\autoref{label_disviluppato}), based on the notion of {\sl canonical partition} of a wall, that guarantees that every possible solution should intersect $S\cup A,$ and we can branch on it.

On the other hand, if the interior of $W$ has few neighbors in $S\cup A$, we find in $W$ a packing of an appropriate number of pairwise disjoint large enough subwalls (\autoref{label_stereotypical}) and we find
a subwall whose interior  has few (a function not depending on $k$) neighbors in $S \cup A$.
We then argue that we can define from it a {\sl flat} wall  in which we can apply the irrelevant vertex technique of Robertson and Seymour~\cite{RobertsonS95XIII} (\autoref{label_mitinbegriffen}).
%
%
We stress that this flat subwall  is not  precisely a subwall
of $W$ but a tiny ``tilt'' of a subwall of $W,$ a concept introduced in~\cite{SauST21amor} that is necessary for our proofs. In order to apply  the irrelevant vertex technique, the main combinatorial tool is \autoref{label_perspicacity}, which as been proved in~\cite{SauST21kapiI} and that is an enhancement of a result of Baste et al.~\cite{BasteST20acom}, as we discuss in \autoref{label_desencajaron}.


\smallskip

In order to achieve the improved running time claimed in \autoref{label_entrepreneur}, we do {\sl not} use iterative compression. Instead, we directly invoke \autoref{label_improvements}, which is a variation of \cite[Lemma 11]{SauST21amor} and whose proof uses~\cite{PerkovicR00anim,AlthausZ19opti,KawarabayashiK20line,AdlerDFST11fast}, that either reports that we have a \no-instance, or concludes that the treewidth of $G$ is polynomially bounded by $k,$ or finds a large wall $W$ in $G.$
If the treewidth is small, we proceed as above. If a large wall is found, we apply \autoref{label_dishonorable} and we now distinguish two cases. If a large flat wall is found, we find an irrelevant vertex using again \autoref{label_appressavamo}. Otherwise, inspired by an idea of Marx and Schlotter~\cite{MarxS07obta}, we exploit the fact that ${\cal G}$ excludes an apex graph, and we use flow techniques to either find a vertex that should belong to the solution, or to conclude that we are dealing with a \no-instance.

\paragraph{Organization of the paper.}
In \autoref{label_reproduisent} we give some basic  definitions and preliminary results. In \autoref{label_forestalling} we introduce  flat walls along with all the concepts and results
around the Flat Wall Theorem, using the framework of \cite{SauST21amor}.
In \autoref{label_compensating} we present several algorithmic and combinatorial results that will be used in the  algorithms, when finding an irrelevant vertex or when applying the branching step. In \autoref{label_pizzighettona} and \autoref{label_demonstrieren} we present the main algorithms claimed in \autoref{label_preliminaries} and \autoref{label_entrepreneur}, respectively.
In \autoref{label_ressentiment} we explain how to modify our algorithms so to deal with a series of variants of the \mnb{\sc Vertex Deletion to ${\cal G}$} problem.
We conclude in \autoref{label_grundgesetzen} with some directions for further research.

\section{Definitions and preliminary results}\label{label_reproduisent}
Our first step is to restate the problem in a more convenient way. We next  give some basic definitions and preliminary results.

\subsection{Restating  the problem}\label{label_interrupting}

Let ${\cal F}$ be a finite non-empty collection of non-empty graphs. We use  ${\cal F}\leq_{\sf m}  G$ to denote that some graph in ${\cal F}$ is a minor of $G.$ 

Given a graph class ${\cal G},$ its {\em minor obstruction set} is defined as the set of all minor-minimal graphs that are not in ${\cal G}$, and is denoted by ${\bf obs}({\cal G}).$
Given a finite non-empty collection of non-empty graphs ${\cal F},$ we denote by ${\bf exc}({\cal F})$ as the set containing every graph ${\cal G}$ that excludes all graphs in ${\cal F}$ as minors.

Let ${\cal G}$ be a minor-closed graph class and ${\cal F}$ be its obstruction set.
Clearly, \mnb{\sc  Vertex Deletion to ${\cal G}$} is the same problem as
asking, given a graph $G$ and some $k\in \Bbb{N},$ for a vertex set $S$ of at most $k$ vertices
such that $G\setminus S\in {\bf exc}({\cal F}).$
Following the terminology of~\cite{BasteST20hittI,BasteST20hittII,BasteST20hittIII,BasteST20acom,FominLPSZ20hitt,FominLMS12plan,KimLPRRSS16line,KimST18data},
we call this problem \mnb{\sc  ${\cal F}$-M-Deletion}.

\paragraph{Some conventions.}
In what follows we always denote by ${\cal F}$
the set ${\bf obs}({\cal G})$ of the instantiation of \mnb{\sc  Vertex Deletion to ${\cal G}$} that we consider.
Notice that, given a graph $G$ and an integer $k,$ $(G,k)$ is a \yes-instance of
\mnb{\sc  ${\cal F}$-M-Deletion} if and only if $G\in {\cal A}_k ({\bf exc}({\cal F})).$
Given a graph $G,$ we define its {\em apex number} to be the smallest integer $a$ for which
$G$ is an $a$-apex of the class of planar graphs.
Also, we define the {\em detail} of $G,$ denoted by ${\sf detail}(G),$
to be the maximum among $|E(G)|$ and $|V(G)|.$
We define three constants depending on ${\cal F}$
that will be used throughout the paper whenever we
consider such a collection ${\cal F}.$ We define  $a_{\cal F}$ as  the minimum
apex number of a graph in ${\cal F},$  we set
$s_{\cal F}=\max\{|V(H)|\mid H\in {\cal F}\},$  and we set $\ell_{\cal F}=\max\{ {\sf detail}(H)|\mid H\in{\cal F}\}.$
Unless stated otherwise, we denote by $n$ and $m$
the number of vertices and edges, respectively, of the graph under consideration.
We can always assume that $G$ has ${\cal O}_{s_{\cal F}}(k\cdot n)$ edges, otherwise we can directly conclude that $(G,k)$ is a \no-instance (for this, use the fact that  graphs excluding some graph as a minor are sparse \cite{Kostochka82lowe,Thomason01thee}).

\subsection{Preliminaries}
\label{label_incurability}

\paragraph{Sets and integers.}\label{label_eruditissime}
We denote by $\Bbb{N}$ the set of non-negative integers.
Given two integers $p$ and $q,$ the set $[p,q]$ contains every integer $r$ such that $p\leq r\leq q.$
For an integer $p\geq 1,$ we set $[p]=[1,p]$ and $\Bbb{N}_{\geq p}=\Bbb{N}\setminus [0,p-1].$
Given a non-negative integer $x,$
we denote by ${\sf odd}(x)$ the minimum odd number that is not smaller than $x.$
For a set $S,$ we denote by $2^{S}$ the set of all subsets of $S$ and, given an integer $r\in[|S|],$
we denote by $\binom{S}{r}$ the set of all subsets of $S$ of size $r$ and by $\binom{S}{\leq r}$ the set of all subsets of $S$ of size at most $r.$
If ${\cal S}$ is a collection of objects where the operation $\cup$ is defined,
then we denote $\cupall {\cal S}=\bigcup_{X\in {\cal S}}X.$

\paragraph{Basic concepts on graphs.}\label{label_shipwrightson}
All graphs considered in this paper are undirected, finite, and without loops or multiple edges.
We use standard graph-theoretic notation and we refer the reader to \cite{Diestel10grap} for any undefined terminology.
Let $G$ be a graph. We say that a pair $(L,R)\in 2^{V(G)}\times 2^{V(G)}$ is a {\em separation} of $G$
if $L\cup R=V(G)$ and there is no edge in $G$ between $L\setminus R$ and $R\setminus L.$
Given a vertex $v\in V(G),$ we denote by $N_{G}(v)$ the set of vertices of $G$ that are adjacent to $v$ in $G.$
A vertex $v \in V(G)$ is \emph{isolated} if $N_G(v) = \emptyset.$
For $S \subseteq V(G),$ we set $G[S]=(S,E\cap{S \choose 2} )$
and use the shortcut $G \setminus S$ to denote $G[V(G) \setminus S].$
Given a vertex $v\in V(G)$ of degree two with neighbors $u$ and $w,$ we define the {\em dissolution} of $v$
to be the operation of deleting $v$ and, if $u$ and $w$ are not adjacent, adding the edge $\{u,w\}.$
Given two graphs $H,G,$ we say that $H$ is a {\em dissolution} of $G$
if $H$ can be obtained from $G$ after dissolving vertices of $G.$
Given an edge $e=\{u,v\}\in E(G),$ we define the {\em subdivision} of $e$
to be the operation of deleting $e,$ adding a new vertex $w$ and making it adjacent to $u$ and $v.$
Given two graphs $H$ and $G,$ we say that $H$ is a {\em subdivision} of $G$
if $H$ can be obtained from $G$ after subdividing edges of $G.$

\paragraph{Treewidth.}
A \emph{tree decomposition} of a graph~$G$
is a pair~$(T,\chi)$ where $T$ is a tree and $\chi: V(T)\to 2^{V(G)}$
such that
\begin{itemize}
	\item $\bigcup_{t \in V(T)} \chi(t) = V(G),$
	\item for every edge~$e$ of~$G$ there is a $t\in V(T)$ such that
	      $\chi(t)$
	      contains both endpoints of~$e,$ and
	\item for every~$v \in V(G),$ the subgraph of~${T}$
	      induced by $\{t \in V(T)\mid {v \in \chi(t)}\}$ is connected.
\end{itemize}
The {\em width} of $(T,\chi)$ is equal to $\max\big\{\left|\chi(t)\right|-1 \bigmid t\in V(T)\big\}$
and the {\em treewidth} of $G$, denoted by $\tw(G)$,  is the minimum width over all tree decompositions of $G.$

To compute a tree decomposition of a graph of bounded treewidth, in the proof of \autoref{label_improvements} in \autoref{label_demonstrieren}
we will use the single-exponential $5$-approximation algorithm for treewidth
of Bodlaender et al. \cite[Theorem VI]{BodlaenderDDFLP16ackn}.

\begin{proposition}\label{label_panathinaiko}
	There is an algorithm that, given an graph $G$ and an integer $k,$
	outputs either a report that $\tw(G)>k,$ or a tree decomposition of $G$ of width at most $5k+4.$
	Moreover, this algorithm runs in $2^{{\cal O}(k)} \cdot n$-time.
\end{proposition}

\paragraph{Contractions and minors.}\label{label_undifferentiated}
The \emph{contraction} of an edge $e = \{u,v\}$ of a simple graph $G$ results in a simple graph $G'$
obtained from $G \setminus \{u,v\}$ by adding a new vertex $uv$ adjacent to all the vertices
in the set $N_G(u) \cup N_G(v)\setminus \{u,v\}.$
A graph $G'$ is a \emph{minor} of a graph $G,$ denoted by $G'\leq_{\sf m}G,$
if $G'$ can be obtained from $G$ by a sequence of vertex removals, edge removals, and edge contractions.
If only edge contractions are allowed, we say that $G'$ is a \emph{contraction} of $G.$
Given two graphs $H$ and $G,$ if $H$ is a minor of $G$ then for every vertex $v\in V(H)$ there is
a set of vertices in $G$ that are the endpoints of the edges of $G$ contracted towards creating $v.$
We call this set {\em model} of $v$ in $G.$
Recall that, given a finite collection of graphs ${\cal F}$ and a graph $G,$
we use notation ${\cal F}\leq_{\sf m} G$ to denote that some graph in ${\cal F}$ is a minor of $G.$

We present here the main result of Baste et al.~\cite{BasteST20acom}, which we will use  in order to solve {\sc ${\cal F}$-M-Deletion} on instances of  treewidth bounded by an appropriate function of $k$.
%

\begin{proposition}\label{label_confiscating}
	Let ${\cal F}$ be a finite collection of graphs.
	There exists an algorithm that, given a triple $(G,\tw,k)$ where $G$ is a graph of treewidth at most $\tw$
	and $k$ is a non-negative integer, it outputs, if it exists, a vertex set $S$ of $G$ of size at most $k$
	such that $G\setminus S\in {\bf exc}({\cal F}).$
	Moreover, this algorithm runs in  $2^{{\cal O}_{s_{\cal F}}(\tw \cdot \log  \tw)}\cdot n$-time.
\end{proposition}

\section{Flat walls}\label{label_forestalling}
In this section we deal with flat walls, using the framework of \cite{SauST21amor}.
More precisely, in \autoref{label_zapateadores}, we introduce walls and several notions concerning them.
In \autoref{label_klammerausdrucks}, we provide the definitions of a rendition and a painting.
Using the above notions, in \autoref{label_souhaiteroient}, we define flat walls and provide some results about them, including the Flat Wall Theorem (namely, the version proved by Kawarabayashi et al. \cite{KawarabayashiTW18anew}) and its algorithmic version restated in the ``more accurate'' framework of  \cite{SauST21amor}.
Finally, in \autoref{label_definitionen}, we present the notion of homogeneity and an algorithm from \cite{SauST21amor} that allows us to detect a homogenous flat wall ``inside'' a given flat wall of ``big enough'' height.
We note that the definitions of this section can also be found in \cite{SauST21amor,SauST21kapiI}.

\subsection{Walls and subwalls}\label{label_zapateadores}
We start with some basic definitions about walls.

\paragraph{Walls.}
Let  $k,r\in\Bbb{N}.$ The
\emph{$(k\times r)$-grid} is the
graph whose vertex set is $[k]\times[r]$ and two vertices $(i,j)$ and $(i',j')$ are adjacent if and only if $|i-i'|+|j-j'|=1.$
An  \emph{elementary $r$-wall}, for some odd integer $r\geq 3,$ is the graph obtained from a
$(2 r\times r)$-grid
with vertices $(x,y)
	\in[2r]\times[r],$
after the removal of the
``vertical'' edges $\{(x,y),(x,y+1)\}$ for odd $x+y,$ and then the removal of
all vertices of degree one.
Notice that, as $r\geq 3,$  an elementary $r$-wall is a planar graph
that has a unique (up to topological isomorphism) embedding in the plane $\Bbb{R}^{2}$
such that all its finite faces are incident to exactly six
edges.
The {\em perimeter} of an elementary $r$-wall is the cycle bounding its infinite face,
while the cycles bounding its finite faces are called {\em bricks}.
Also, the vertices
in the perimeter of an elementary $r$-wall that have degree two are called {\em pegs},
while the vertices $(1,1), (2,r), (2r-1,1), (2r,r)$ are called {\em corners} (notice that the corners are also pegs).

\begin{figure}[h]
	\begin{center}
		\includegraphics[width=11cm]{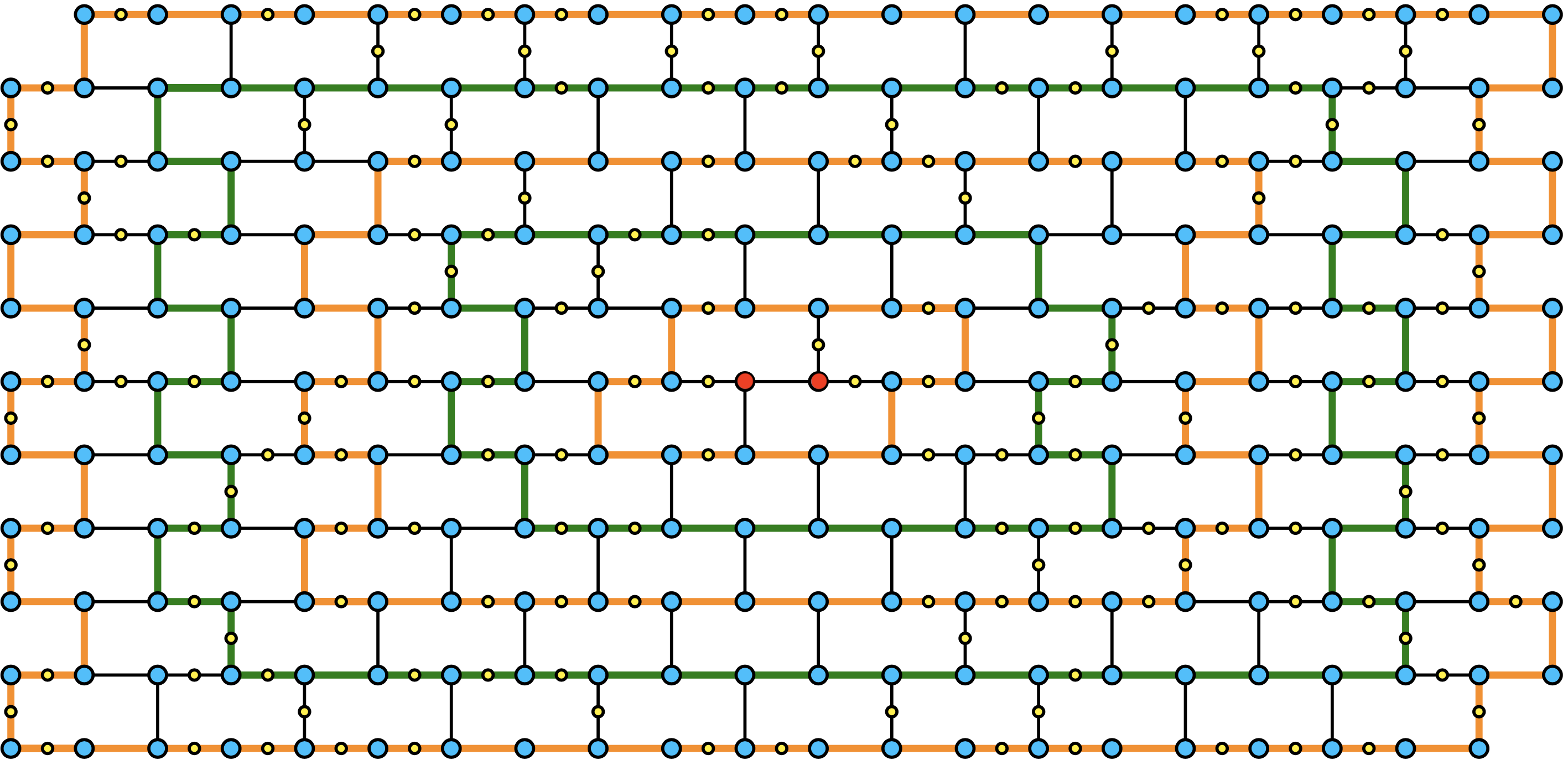}
	\end{center}
	\caption{An $11$-wall and its five layers, depicted in alternating orange and green. The central vertices of the wall are depicted in red.}
	\label{label_estatuderato}
\end{figure}

An {\em $r$-wall} is any graph $W$ obtained from an elementary $r$-wall $\bar{W}$
after subdividing edges (see \autoref{label_estatuderato}). A graph $W$ is a {\em wall} if it is an $r$-wall for some odd $r\geq 3$
and we refer to $r$ as the {\em height} of $W.$ Given a graph $G,$
a {\em wall of} $G$ is a subgraph of $G$ that is a wall.
We insist that, for every $r$-wall, the number $r$ is always odd.

We call the vertices of degree three of a wall $W$ {\em 3-branch vertices}.
A cycle of $W$ is a {\em brick} (resp. the {\em perimeter}) of $W$
if its 3-branch vertices are the vertices of a brick (resp. the perimeter) of $\bar{W}.$
We denote by ${\cal C}(W)$ the set of all cycles of $W.$
We  use $D(W)$ in order to denote the perimeter of the  wall $W.$
A brick of $W$ is {\em internal} if it is disjoint from $D(W).$

\paragraph{Subwalls.}
Given an elementary $r$-wall $\bar{W},$ some odd $i\in \{1,3,\ldots,2r-1\},$ and $i'=(i+1)/2,$
the {\em $i'$-th  vertical path} of $\bar{W}$  is the one whose
vertices, in order of appearance, are $(i,1),(i,2),(i+1,2),(i+1,3),
	(i,3),(i,4),(i+1,4),(i+1,5),
	(i,5),\ldots,(i,r-2),(i,r-1),(i+1,r-1),(i+1,r).$
Also, given some $j\in[2,r-1]$ the {\em $j$-th horizontal path} of $\bar{W}$
is the one whose
vertices, in order of appearance, are $(1,j),(2,j),\ldots,(2r,j).$

A \emph{vertical} (resp. \emph{horizontal}) path of an $r$-wall $W$ is one
that is a subdivision of a  vertical (resp. horizontal) path of $\bar{W}.$
Notice that the perimeter of an $r$-wall $W$
is uniquely defined regardless of the choice of the elementary $r$-wall $\bar{W}.$
A {\em subwall} of $W$ is any subgraph $W'$ of  $W$
that is an $r'$-wall, with $r' \leq r,$ and such the vertical (resp. horizontal) paths of $W'$ are subpaths of the
	{vertical} (resp. {horizontal}) paths of $W.$

\paragraph{Layers.}
The {\em layers} of an $r$-wall $W$  are recursively defined as follows.
The first layer of $W$ is its perimeter.
For $i=2,\ldots,(r-1)/2,$ the $i$-th layer of $W$ is the $(i-1)$-th layer of the subwall $W'$
obtained from $W$ after removing from $W$ its perimeter and
removing recursively all occurring vertices of degree one.
We refer to the $(r-1)/2$-th layer as the {\em inner layer} of $W.$
The {\em central vertices} of an $r$-wall are its two branch vertices  that do not belong to any of its layers.
See \autoref{label_estatuderato} for an illustration of the notions defined above.

\paragraph{Central walls.}
Given an $r$-wall $W$ and an odd $q\in\Bbb{N}_{\geq 3}$ where $q\leq r,$
we define the {\em central $q$-subwall} of $W,$ denoted by $W^{(q)},$
to be the $q$-wall obtained from $W$ after removing
its first $(r-q)/2$ layers and all occurring vertices of degree one.

\paragraph{Tilts.}
The {\em interior} of a wall $W$ is the graph obtained
from $W$ if we remove from it all edges of $D(W)$ and all vertices
of $D(W)$ that have degree two in $W.$
Given two walls $W$ and $\tilde{W}$ of a graph $G,$
we say that $\tilde{W}$ is a {\em tilt} of $W$ if $\tilde{W}$ and $W$ have identical interiors.
\medskip

The following result is derived from \cite{AdlerDFST11fast}. We will use it in the improved algorithm of \autoref{label_entrepreneur} in \autoref{label_demonstrieren}, in order to find a wall in a graph of bounded treewidth, given a tree decomposition of it.
\begin{proposition}\label{label_sleepwalkers}
	There is an algorithm that, given a graph $G$ on $m$ edges,
	a graph $H$ on $h$ edges  without isolated vertices,
	and a tree decomposition of $G$ of width at most $k,$ it outputs,
	if it exists, a minor of $G$ isomorphic to $H.$
	Moreover, this algorithm runs in $2^{{\cal O}(k\log k)}\cdot h^{{\cal O}(k)}\cdot 2^{{\cal O}(h)}\cdot m$-time.
\end{proposition}

\subsection{Paintings and renditions}
\label{label_klammerausdrucks}
In this subsection we present the notions of renditions and paintings, originating in the work of Robertson and Seymour \cite{RobertsonS95XIII}.
The definitions presented here were introduced by Kawarabayashi et al. \cite{KawarabayashiTW18anew} (see also \cite{BasteST20acom,SauST21amor}).
\paragraph{Paintings.}
A {\em closed} (resp. {\em open}) {\em disk} is a set homeomorphic to the set
$\{(x,y)\in \Bbb{R}^{2}\mid x^{2}+y^{2}\leq 1\}$ (resp. $\{(x,y)\in \Bbb{R}^{2}\mid x^{2}+y^{2}< 1\}$).
Let $\Delta$ be a closed disk.
Given a subset $X$ of $\Delta,$ we
denote its closure by $\bar{X}$ and its boundary by $\bd(X).$
A {\em {$\Delta$}-painting} is a pair $\Gamma=(U,N)$
where
\begin{itemize}
	\item  $N$ is a finite set of points of $\Delta,$
	\item $N \subseteq U \subseteq \Delta,$ and
	\item $U \setminus  N$ has finitely many arcwise-connected  components, called {\em cells}, where, for every cell $c,$
	      \begin{itemize}
		      \item[$\circ$] the closure $\bar{c}$ of $c$
		            is a closed disk
		            and
		      \item[$\circ$]  $|\tilde{c}|\leq 3,$ where $\tilde{c}:=\bd(c)\cap N.$
	      \end{itemize}
\end{itemize}
We use the  notation $U(\Gamma) := U,$
$N(\Gamma) := N$  and denote the set of cells of $\Gamma$
by $C(\Gamma).$
For convenience, we may assume that each cell  of $\Gamma$ is an open disk of $\Delta.$

Notice that, given a $\Delta$-painting $\Gamma,$
the pair $(N(\Gamma),\{\tilde{c}\mid c\in C(\Gamma)\})$  is a hypergraph whose hyperedges have cardinality at most three and  $\Gamma$ can be seen as a plane embedding of this hypergraph in $\Delta.$

\paragraph{Renditions.}
Let $G$ be a graph and let $\Omega$ be a cyclic permutation of a subset of $V(G)$ that we denote by $V(\Omega).$ By an {\em $\Omega$-rendition} of $G$ we mean a triple $(\Gamma, \sigma, \pi),$ where
\begin{itemize}
	\item[(a)] $\Gamma$ is a $\Delta$-painting for some closed disk $\Delta,$
	\item[(b)] $\pi: N(\Gamma)\to V(G)$ is an injection, and
	\item[(c)] $\sigma$ assigns to each cell $c \in  C(\Gamma)$ a subgraph $\sigma(c)$ of $G,$ such that
	      \begin{enumerate}
		      \item[(1)] $G=\bigcup_{c\in C(\Gamma)}\sigma(c),$
		      \item[(2)]  for distinct $c, c' \in  C(\Gamma),$  $\sigma(c)$ and $\sigma(c')$  are edge-disjoint,
		      \item[(3)] for every cell $c \in  C(\Gamma),$ $\pi(\tilde{c}) \subseteq V (\sigma(c)),$
		      \item[(4)]  for every cell $c \in  C(\Gamma),$
		            $V(\sigma(c)) \cap \bigcup_{c' \in  C(\Gamma) \setminus  \{c\}}V(\sigma(c')) \subseteq \pi(\tilde{c}),$ and
		      \item[(5)]  $\pi(N(\Gamma)\cap \bd(\Delta))=V(\Omega),$ such that the points
		            in $N(\Gamma)\cap \bd(\Delta)$ appear in $\bd(\Delta)$ in the same ordering
		            as their images, via $\pi,$ in $\Omega.$
	      \end{enumerate}
\end{itemize}

\subsection{Flatness pairs}
\label{label_souhaiteroient}
In this subsection we define the notion of a flat wall.
The definitions given in this subsection are originating in \cite{SauST21amor}.
We refer the reader to that paper for a more detailed exposition of these definitions and the reasons for which we introduced them.
We use the more accurate framework of \cite{SauST21amor} concerning flat walls, instead of that of \cite{KawarabayashiTW18anew},
in order to be able to use tools that are developed in \cite{SauST21amor} and \cite{SauST21kapiI} and will be useful in future applications as well.
\paragraph{Flat walls.}
Let $G$ be a graph and let $W$ be an $r$-wall  of $G,$ for some odd integer $r\geq 3.$
We say that a pair $(P,C)\subseteq D(W)\times D(W)$ is a {\em choice
		of pegs and corners for $W$} if $W$ is the subdivision of an  elementary $r$-wall $\bar{W}$
where $P$ and
$C$ are the pegs and the corners of $\bar{W},$ respectively (clearly, $C\subseteq P$).
To get more intuition, notice that a wall $W$ can occur in several ways from the elementary wall $\bar{W},$
depending on the way the vertices in the perimeter of $\bar{W}$ are subdivided.
Each of them gives a different selection $(P,C)$ of pegs and corners of $W.$

We say that $W$ is a {\em flat $r$-wall}
of $G$ if there is a separation $(X,Y)$ of $G$ and a choice  $(P,C)$
of pegs and corners for $W$ such that:
\begin{itemize}
	\item $V(W)\subseteq Y,$
	\item  $P\subseteq X\cap Y\subseteq V(D(W)),$ and
	\item  if $\Omega$ is the cyclic ordering of the vertices $X\cap Y$ as they appear in $D(W),$
	      then there exists an $\Omega$-rendition $(\Gamma,\sigma,\pi)$ of  $G[Y].$
\end{itemize}
We say that $W$ is a {\em flat wall}
of $G$ if it is a flat $r$-wall for some odd integer $r \geq 3.$

\paragraph{Flatness pairs.}
Given the above, we  say that  the choice of the 7-tuple $\frak{R}=(X,Y,P,C,\Gamma,\sigma,\pi)$
{\em certifies that $W$ is a flat wall of $G$}.
We call the pair $(W,\frak{R})$ a {\em flatness pair} of $G$ and define
the {\em height} of the pair $(W,\frak{R})$ to be the height of $W.$
We use the term {\em cell of} $\frak{R}$ in order to refer to the cells of $\Gamma.$

We call the graph $G[Y]$ the {\em $\frak{R}$-compass} of $W$ in $G,$
denoted by ${\sf compass}_{\frak{R}}(W).$
We can assume that ${\sf compass}_{\frR} (W)$ is connected, updating $\frR$ by removing from $Y$ the vertices of all the connected components of ${\sf compass}_\frR (W)$
except of the one that contains $W$ and including them in $X$ ($\Gamma, \sigma, \pi$ can also be easily modified according to the removal of the aforementioned vertices from $Y$).
We define the  {\em flaps} of the wall $W$ in $\frak{R}$ as
${\sf flaps}_{\frak{R}}(W):=\{\sigma(c)\mid c\in C(\Gamma)\}.$
Given a flap $F\in {\sf flaps}_{\frak{R}}(W),$ we define its {\em base}
as $\partial F:=V(F)\cap \pi(N(\Gamma)).$
A  cell $c$ of ${\frR}$ is {\em untidy} if  $\pi(\tilde{c})$ contains a vertex
$x$ of ${W}$ such that two of the edges of ${W}$ that are incident to $x$ are edges of $\sigma(c).$ Notice that if $c$ is untidy then  $|\tilde{c}|=3.$
A cell $c$ of $\frR$ is {\em tidy} if it is not untidy.

\paragraph{Cell classification.}
Given a cycle $C$ of ${\sf compass}_{\frak{R}}(W),$ we say that
$C$ is {\em $\frak{R}$-normal} if it is {\sl not} a subgraph of a flap $F\in {\sf flaps}_{\frak{R}}(W).$
Given an $\frak{R}$-normal cycle $C$ of ${\sf compass}_{\frak{R}}(W),$
we call a cell $c$ of $\frak{R}$ {\em $C$-perimetric} if
$\sigma(c)$ contains some edge of $C.$
Notice that if $c$ is $C$-perimetric, then $\pi(\tilde{c})$ contains two points $p,q\in N(\Gamma)$
such that  $\pi(p)$ and $\pi(q)$ are vertices of $C$ where one,
say $P_{c}^{\rm in},$ of the two $(\pi(p),\pi(q))$-subpaths of $C$ is a subgraph of $\sigma(c)$ and the other,
denoted by $P_{c}^{\rm out},$  $(\pi(p),\pi(q))$-subpath contains at most one internal vertex of $\sigma(c),$
which should be the (unique) vertex $z$ in $\partial\sigma(c)\setminus\{\pi(p),\pi(q)\}.$
We pick a $(p,q)$-arc $A_{c}$ in $\hat{c}:={c}\cup\tilde{c}$ such that  $\pi^{-1}(z)\in A_{c}$ if and only if $P_{c}^{\rm in}$ contains
the vertex $z$ as an internal vertex.

We consider the circle  $K_{C}=\cupall\{A_{c}\mid \mbox{$c$ is a $C$-perimetric cell of $\frak{R}$}\}$
and we denote by $\Delta_{C}$ the closed disk bounded by $K_{C}$  that is contained in  $\Delta.$
A cell $c$ of $\frak{R}$ is called {\em $C$-internal} if $c\subseteq \Delta_{C}$
and is called {\em $C$-external} if $\Delta_{C}\cap c=\emptyset.$
Notice that  the cells of $\frak{R}$ are partitioned into  $C$-internal,  $C$-perimetric, and  $C$-external cells.

Let $c$ be a tidy $C$-perimetric cell of $\frak{R}$ where $|\tilde{c}|=3.$ Notice that $c\setminus A_{c}$ has two arcwise-connected components and one of them is an open disk $D_{c}$ that is a subset of $\Delta_{C}.$
If the closure $\overline{D}_{c}$  of $D_{c}$ contains only two points of $\tilde{c}$ then we call the cell $c$ {\em $C$-marginal}.

\paragraph{Influence.}
For every $\frak{R}$-normal cycle $C$ of ${\sf compass}_{\frak{R}}(W)$ we define the set
$${\sf influence}_{\frak{R}}(C)=\{\sigma(c)\mid \mbox{$c$ is a cell of $\frak{R}$ that is not $C$-external}\}.$$
%
%

A wall $W'$  of ${\sf compass}_{\frak{R}}(W)$  is \emph{$\frak{R}$-normal} if $D(W')$ is  $\frak{R}$-normal.
Notice that every wall of $W$ (and hence every subwall of $W$) is an $\frak{R}$-normal wall of ${\sf compass}_{\frak{R}}(W).$ We denote by ${\cal S}_{\frak{R}}(W)$ the set of all $\frak{R}$-normal walls of ${\sf compass}_{\frak{R}}(W).$ Given a wall $W'\in {\cal S}_{\frak{R}}(W)$ and a cell $c$ of $\frak{R}$,
we say that $c$ is {\em $W'$-perimetric/internal/external/marginal} if $c$ is  $D(W')$-perimetric/internal/external/marginal, respectively.
We also use $K_{W'},$ $\Delta_{W'},$ ${\sf influence}_{\frak{R}}(W')$ as shortcuts
for $K_{D(W')},$ $\Delta_{D(W')},$ ${\sf influence}_{\frak{R}}(D(W')),$ respectively.

\paragraph{Regular flatness pairs.}
We call a  flatness pair $(W,\frak{R})$ of a graph $G$ {\em regular}
if none of its cells is $W$-external, $W$-marginal, or untidy.

\paragraph{Tilts of flatness pairs.}
Let $(W,\frak{R})$ and $(\tilde{W}',\tilde{\frak{R}}')$  be two flatness pairs of a graph $G$
and let $W'\in {\cal S}_{\frak{R}}(W).$
We assume that ${\frak{R}}=(X,Y,P,C,\Gamma,\sigma,\pi)$
and $\tilde{\frak{R}}'=(X',Y',P',C',\Gamma',\sigma',\pi').$
We say that   $(\tilde{W}',\tilde{\frak{R}}')$   is a {\em $W'$-tilt}
of $(W,\frak{R})$ if
\begin{itemize}
	\item $\tilde{\frak{R}}'$ does not have $\tilde{W}'$-external cells,
	\item  $\tilde{W}'$ is a tilt of $W',$
	\item  the set of $\tilde{W}'$-internal  cells of  $\tilde{\frak{R}}'$ is the same as the set of $W'$-internal
	      cells of ${\frak{R}}$ and their images via $\sigma'$ and ${\sigma}$ are also the same,
	\item ${\sf compass}_{\tilde{\frak{R}}'}(\tilde{W}')$ is a subgraph of $\cupall{\sf influence}_{{\frak{R}}}(W'),$ and
	\item if $c$ is a cell in $C(\Gamma') \setminus C(\Gamma),$ then $|\tilde{c}| \leq 2.$
\end{itemize}

The next observation follows from the third item above and the fact that the cells corresponding to flaps
containing a central vertex of $W'$ are all internal (recall that the height of a wall is always at least three).

\begin{observation}\label{label_stepdaughter}
	Let $(W,\frR)$ be a flatness pair of a graph $G$ and $W'\in{\cal S}_{\frR}(W).$
	For every $W'$-tilt $(\tilde{W}',\tilde{\frR}')$ of $(W,\frR),$ the central vertices of $W'$ belong to the vertex set of ${\sf compass}_{\tilde{\frR}'}(\tilde{W}').$
\end{observation}

Also, given a regular flatness pair $(W,\frR)$ of a graph $G$ and a $W'\in {\cal S}_{\frak{R}}(W),$
for every $W'$-tilt $(\tilde{W}', \tilde{\frR}')$ of $(W,\frak{R}),$ by definition none of its cells is $\tilde{W}'$-external, $\tilde{W}'$-marginal, or untidy -- thus, $(\tilde{W}', \tilde{\frR}')$ is regular.
Therefore, regularity of a flatness pair is a property that its tilts ``inherit''.

\begin{observation}\label{label_ressemblances}
	If $(W,\frak{R})$ is a regular flatness pair, then for every $W'\in {\cal S}_{\frak{R}}(W),$ every $W'$-tilt of $(W,\frak{R})$ is also regular.
\end{observation}

We next present one of the main results of \cite{SauST21amor}.

\begin{proposition}
	\label{label_protectively}
	There exists an algorithm that, given a graph $G,$ a flatness pair $({W},{\frak{R}})$ of $G,$ and a wall $W'\in {\cal S}_{\frak{R}}(W),$ outputs  a  $W'$-tilt of $({W},{\frak{R}})$ in  ${\cal O}(n+m)$-time.
\end{proposition}
%

We present here the Flat Wall Theorem and, in particular, the version proved by Kawarabayashi et al.~\cite[Theorem 1.5]{KawarabayashiTW18anew}.
This result will be used in the proof of correctness of the algorithm of \autoref{label_entrepreneur}.

\begin{proposition}\label{label_propositional}
	There are two functions  $\newfun{label_discernimiento}:\Bbb{N}\to \Bbb{N}$  and
	$\newfun{label_schematization}:\Bbb{N}\to \Bbb{N},$ where the images of $\funref{label_discernimiento}$ are odd numbers, such that if $r$ is an odd integer in $\Bbb{N}_{\geq 3},$ $t\in\Bbb{N}_{\geq 1},$
	$G$ is a graph that does not contain $K_t$ as a minor,  and  $W$ is an $\funref{label_discernimiento}(t)\cdot r$-wall of $G,$
	then there is a set $A\subseteq V(G)$ where $|A|\leq \funref{label_schematization}(t)$
	and a flatness pair $(\tilde{W}',\tilde{\frak{R}}')$ of $G\setminus A$ of height $r.$
	Moreover $\funref{label_discernimiento}(t)={\cal O}(t^{26})$ and $\funref{label_schematization}(t)={\cal O}(t^{24}).$
\end{proposition}

We conclude this subsection with the following result from \cite{SauST21amor} that allows us to find a regular flatness pair in a minor-free graph of ``big enough'' treewidth.

\begin{proposition}\label{label_dishonorable}
	There is a function   $\newfun{label_inconsummabile}:\Bbb{N}\to \Bbb{N}$    and
	an algorithm that receives as  input  a graph $G,$ an odd integer $r\geq 3,$ and a  $t\in\Bbb{N}_{\geq 1},$ and  outputs, in time $2^{{\cal O}_{t}(r^2)}\cdot n$, one of the following:
	\begin{itemize}
		\item a report  that $K_{t}$ is a minor of $G,$
		\item a tree decomposition of $G$ of width at most $\funref{label_inconsummabile}(t)\cdot r,$ or
		\item a set $A\subseteq V(G)$ with $|A|\leq \funref{label_schematization}(t)$ and a regular flatness pair $(W,\frak{R})$ of $G\setminus A$ of height $r.$
		      (Here $\funref{label_schematization}(t)$ is the function of \autoref{label_propositional}.)
	\end{itemize}
	Moreover, $\funref{label_inconsummabile}(t)=2^{{\cal O}(t^2 \log t)}.$
\end{proposition}

We note that the result of \cite{SauST21amor} also returns a tree decomposition of the flatness pair. However, this additional output is not needed in the algorithms of this paper.

\subsection{Homogeneous walls}\label{label_definitionen}
We first present some definitions on boundaried graphs and folios that will be used to define the notion of homogeneous walls.
Following this, we present some results concerning homogeneous walls that are key ingredients in the application of the irrelevant vertex technique in our proofs.

\paragraph{Boundaried graphs.}
Let $t\in\Bbb{N}.$
A \emph{$t$-boundaried graph} is a triple $\bound{G} = (G,B,\rho)$ where $G$ is a graph, $B \subseteq V(G),$ $|B| = t,$ and
$\rho : B \to [t]$ is a bijection.
We  say that  $\textbf{G}_1=(G_1,B_1,\rho_1)$ and $\textbf{G}_{2}=(G_2,B_2,\rho_2)$
are {\em isomorphic} if there is an isomorphism from $G_{1}$ to $G_{2}$
that extends the bijection $\rho_{2}^{-1}\circ \rho_{1}.$
The triple $(G,B,\rho)$ is a {\em boundaried graph} if it is a $t$-boundaried graph for some $t\in\Bbb{N}.$
As in~\cite{RobertsonS95XIII} (see also \cite{BasteST20acom}), we define the {\em detail} of a boundaried graph
$\bound{G} = (G,B,\rho)$ as  ${\sf detail}(\bound{G}):=\max\{|E(G)|,|V(G)\setminus B|\}.$
We denote by ${\cal B}^{(t)}$ the set of all (pairwise non-isomorphic)  $t$-boundaried graphs and by ${\cal B}_{\ell}^{(t)}$ the set of all (pairwise non-isomorphic) $t$-boundaried graphs with detail at most $\ell.$
We also set ${\cal B}=\bigcup_{t\in\Bbb{N}}{\cal B}^{(t)}.$

We define the {\em treewidth} of a boundaried graph ${\bf G}=(G,B,\rho),$ denoted by $\tw({\bf G}),$ as the minimum width of a tree decomposition $(T,\chi)$ of $G$ for which there is some $u\in V(T)$ such that $B\subseteq \chi(u).$
Notice that the treewidth of a $t$-boundaried graph is always lower-bounded by $t-1.$

\paragraph{Folios.}
We say that $(M,T)$   is a {\em {\sf tm}-pair} if $M$ is  a graph, $T\subseteq V(M),$ and  all vertices in
$V(M)\setminus T$ have degree two. We denote by ${\sf diss}(M,T)$ the graph obtained
from  $M$ by {dissolving} all vertices  in $V(M)\setminus T.$
A {\em {\sf tm}-pair} of a graph $G$  is a  {\em {\sf tm}-pair}  $(M,T)$ where $M$ is a subgraph of $G.$
We call the vertices in $T$ {\em branch} vertices of $(M,T).$
We need to deal with topological minors for the notion of homogeneity defined below, on which  the statement of~\cite[Theorem 5.2]{BasteST20acom} relies.
If $\textbf{M}=(M,B,\rho)\in{\cal B}$ and   $T\subseteq V(M)$ with $B\subseteq T,$ we  call  $(\textbf{M},T)$ a {\em {\sf btm}-pair}
and we  define  ${\sf diss}(\textbf{M},T)=({\sf diss}(M, T),B,\rho).$ Note that we do not permit dissolution of boundary vertices, as we consider all of them to be branch vertices. If $\textbf{G}=(G,B,\rho)$ is a boundaried graph and $(M,T)$ is a  {\sf tm}-pair of $G$
where $B\subseteq T,$  then we say that
$(\textbf{M},T),$ where $\textbf{M}=(M,B,\rho),$ is a   {\em {\sf btm}-pair} of $\textbf{G}=(G,B,\rho).$
Let $\textbf{G}_{1},{\bf G}_{2}$ be two boundaried graphs.
We say that $\textbf{G}_{1}$ is a {\em topological minor}
of $\textbf{G}_{2},$ denoted by $\textbf{G}_{1}\pretp\textbf{G}_{2},$ if
$\textbf{G}_{2}$ has a {\sf btm}-pair $(\textbf{M},T)$
such that  ${\sf diss}(\textbf{M},T)$ is isomorphic to $\textbf{G}_{1}.$
Given a $\textbf{G}\in {\cal B} $ and a positive integer $\ell$, we define the {\em $\ell$-folio} of ${\bf G}$
as
$${\ell}\mbox{\sf-folio}(\textbf{G})=\{\textbf{G}'\in {\cal  B} \mid \textbf{G}'\pretp \textbf{G} \mbox{~and $\textbf{G}'$ has detail at most $\ell$}\}.$$

The number of distinct $\ell$-folios of $t$-boundaried graphs is upper-bounded in the following result, proved first in~\cite{BasteST20hittI} and used also in~\cite{BasteST20acom}.
\begin{proposition}\label{label_surprendront}
	There exists a function $\newfun{label_nompareilles}: \Bbb{N}^{2} \to \Bbb{N}$ such that for every $t,\ell\in \Bbb{N},$ $|\{\ell\mbox{\sf-folio}({\bf G}) \mid {\bf G}\in {\cal B}_{\ell}^{(t)}\}|\leq \funref{label_nompareilles}(t,\ell).$ Moreover, $\funref{label_nompareilles}(t,\ell)=2^{2^{{\cal O}((t+\ell)\cdot\log(t+\ell))}}.$\end{proposition}

\paragraph{Augmented flaps.}
Let $G$ be a graph, $A$ be a subset of $V(G)$ of size $a,$ and $(W,\frR)$ be a flatness pair of $G\setminus A.$
For each flap $F\in {\sf flaps}_{\frak{R}}(W)$ we consider a labeling  $\ell_{F}: \partial F\rightarrow\{1,2,3\}$ such that
the set of labels assigned by $\ell_{F}$ to $\partial F$ is  one of $\{1\},$ $\{1,2\},$ $\{1,2,3\}.$
Also, let $\tilde{a}\in[a].$
For every set $\tilde{A}\in\binom{A}{\tilde{a}},$ we consider a bijection  $\rho_{\tilde{A}}: \tilde{A}\to [\tilde{a}].$
The labelings in ${\cal L}=\{\ell_{F} \mid F\in  {\sf flaps}_{\frak{R}}(W)\}$ and the labelings in $\{\rho_{\tilde{A}} \mid \tilde{A}\in\binom{A}{\tilde{a}}\}$ will be useful for defining a set of boundaried graphs that we will call augmented flaps.
We first need some more definitions.

Given a flap $F\in{\sf flaps}_{\frak{R}}(W),$ we define an ordering
$\Omega(F)=(x_{1},\ldots,x_{q}),$ with $q\leq 3,$ of the vertices of $\partial{F}$
so that
\begin{itemize}
	\item $(x_{1},\ldots,x_{q})$ is a  counter-clockwise cyclic ordering of the vertices of $\partial F$ as they appear in the corresponding cell of $C(\Gamma).$ Notice that this cyclic ordering is significant  only when $|\partial F|=3,$
	      in the sense that $(x_{1},x_{2},x_{3})$ remains invariant under shifting, i.e., $(x_{1},x_{2},x_{3})$ is the same as $ (x_{2},x_{3},x_{1})$ but not  under inversion, i.e.,   $(x_{1},x_{2},x_{3})$ is not the same as $(x_{3},x_{2},x_{1}),$ and
	\item   for $i\in[q],$ $\ell_{F}(x_{i})=i.$
\end{itemize}
Notice that the second condition is necessary for completing the definition of the ordering $\Omega(F),$
and this is the reason why we set up the labelings in ${\cal L}.$\medskip\medskip

For each set $\tilde{A}\in\binom{A}{\tilde{a}}$ and each $F\in {\sf flaps}_{\frak{R}}(W)$ with $t_{F}:=|\partial F|,$
we fix $\rho_{F}: \partial F\to [\tilde{a}+1,\tilde{a}+t_F]$ such that
$(\rho^{-1}_{F}(\tilde{a}+1),\ldots,\rho^{-1}_{F}(\tilde{a}+t_F))= \Omega(F).$
Also, we define the boundaried graph $$\textbf{F}^{\tilde{A}}:=(G[\tilde{A}\cup F],\tilde{A}\cup \partial F,\rho_{\tilde{A}}\cup \rho_F)$$
and we denote by $F^{\tilde{A}}$ the underlying graph of $\textbf{F}^{\tilde{A}}.$ We call $\textbf{F}^{\tilde{A}}$ an {\em $\tilde{A}$-augmented flap} of the flatness pair $(W,\frak{R})$ of $G\setminus A$
in $G.$
\paragraph{Palettes and homogeneity.}
For each $\frR$-normal cycle $C$ of ${\sf compass}_\frR (W)$ and each set $\tilde{A}\subseteq A,$ we define $(\tilde{A},\ell)\mbox{\sf -palette}(C)=\{\ell\mbox{\sf -folio}({\bf F}^{\tilde{A}})\mid F\in {\sf  influence}_{\frak{R}}(C)\}.$
Given a set $\tilde{A}\subseteq A,$ we say that the flatness pair $(W,\frak{R})$  of $G\setminus A$ is {\em $\ell$-homogeneous with respect to $\tilde{A}$} if every  {\sl internal} brick of ${W}$ has the {\sl same} $(\tilde{A},\ell)$\mbox{\sf -palette} (seen as a cycle of ${\sf compass}_\frR (W)$).
Also, given a collection ${\cal S}\subseteq 2^A,$ we say that the flatness pair $(W,\frak{R})$  of $G\setminus A$ is {\em $\ell$-homogeneous
		with respect to ${\cal S}$}
if it is $\ell$-homogeneous with respect to every $\tilde{A}\in {\cal S}.$

The following observation is a consequence of the fact that, given a wall $W$ and a  subwall $W'$ of $W,$ every internal brick of a tilt $W''$ of $W'$ is also an internal brick of $W.$

\begin{observation}\label{label_superintendent}
	Let $\ell\in\Bbb{N},$ $G$ be a graph, $A\subseteq V(G),$ ${\cal S}\subseteq 2^A,$ and $(W,\frak{R})$  be a flatness pair of $G\setminus A.$ If $(W,\frak{R})$ is  $\ell$-homogeneous
	with respect to ${\cal S},$ then for every subwall $W'$ of $W,$ every $W'$-tilt of $(W,\frak{R})$ is also $\ell$-homogeneous
	with respect to ${\cal S}.$
\end{observation}

\medskip
%

Let $a,\tilde{a},\ell\in \Bbb{N},$ where $\tilde{a}\leq a.$
Also, let $G$ be a graph, $A$ be a subset of $V(G)$ of size at most $a,$ and $(W,\frR)$ be a flatness pair of $G\setminus A.$
For every flap  $F\in{\sf flaps}_\frR (W),$ we define the function
${\sf var}^{(A,\tilde{a},\ell)}_F:\binom{A}{\leq \tilde{a}}\to \{\ell\mbox{\sf-folio}({\bf G}) \mid {\bf G}\in \bigcup_{i\in[\tilde{a}+3]}{\cal B}^{(i)}\}$
that maps each set $\tilde{A}\in\binom{A}{\leq \tilde{a}}$ to the set  $\ell\mbox{\sf -folio}({\bf F}^{\tilde{A}}).$

{We also use the following result that follows from \autoref{label_surprendront} and the fact that $|\binom{A}{\leq \tilde{a}}|= {\cal O}(|A|^{\tilde{a}})$ (see also \cite{SauST21kapiI}).}

\begin{lemma}\label{label_independents}
	There exists a function $\newfun{label_irresistibility}:\Bbb{N}^3\to \Bbb{N}$ such that if $a,\tilde{a},\ell\in \Bbb{N},$ where $\tilde{a}\leq a,$ $G$ is a graph, $A$ is a subset of $V(G)$ of size at most $a,$ and $(W,\frR)$ is a flatness pair of $G\setminus A,$ then $$|\{{\sf var}^{(A,\tilde{a},\ell)}_F\mid F\in {\sf flaps}_\frR (W)\}|\leq \funref{label_irresistibility}(a,\tilde{a},\ell).$$
	Moreover, $ \funref{label_irresistibility}(a,\tilde{a},\ell)= 2^{a^{\tilde{a}} \cdot 2^{{\cal O}((\tilde{a}+\ell )\cdot \log (\tilde{a}+\ell))}}.$
\end{lemma}

\autoref{label_independents} allows us to define an injective function $\sigma: \{{\sf var}^{(A,\tilde{a},\ell)}_F\mid F\in {\sf flaps}_\frR (W)\} \to [\funref{label_irresistibility}(a,\tilde{a},\ell)]$
that maps each function in $\{{\sf var}^{(A,\tilde{a},\ell)}_F\mid F\in {\sf flaps}_\frR (W)\}$ to an integer in $[\funref{label_irresistibility}(a,\tilde{a},\ell)].$
Using $\sigma,$ we define a function $\zeta_{A,\tilde{a},\ell}:{\sf flaps}_\frR (W)\to [\funref{label_irresistibility}(a,\tilde{a},\ell)],$ that maps each flap $F\in{\sf flaps}_\frR (W)$ to the integer $\sigma({\sf var}^{(A,\tilde{a},\ell)}_F).$
In \cite{SauST21amor}, given a $w\in\Bbb{N},$ the notion of homogeneity is defined with respect to a {\sl flap-coloring $\zeta$ of $(W,\frR)$ with $w$ colors}, that is a function from ${\sf flaps}_\frR (W)$ to $[w].$
This function gives rise to the $\zeta\mbox{\sf -palette}$ of each $\frR$-normal cycle of ${\sf compass}_{\frR} (W)$ which, in turn, is used to define the notion of a {\em $\zeta$-homogeneous} flatness pair.
Hence, using the terminology of \cite{SauST21amor},  $\zeta_{A,\tilde{a},\ell}$ is a flap-coloring of $(W,\frR)$ with $\funref{label_irresistibility}(a,\tilde{a},\ell)$ colors, that ``colors''  each flap $F\in {\sf flaps}_\frR (W)$ by mapping it to the integer $\sigma({\sf var}^{(A,\tilde{a},\ell)}_F)$, and the notion of $\ell$-homogeneity
with respect to ${A \choose \leq {\tilde{a}}}$ defined here can be alternatively interpreted as $\zeta_{A,\tilde{a},\ell}$-homogeneity.
The following result, which is the application of a result of Sau et al.~\cite[Lemma 13]{SauST21amor} for the flap-coloring $\zeta_{A,\tilde{a},\ell},$ provides an algorithm that, given a flatness pair of ``big enough'' height, outputs a homogeneous flatness pair.

\begin{proposition}\label{label_highlighting}
	There is a function $\newfun{label_distinguimos}:\Bbb{N}^4\to \Bbb{N},$ whose images are odd integers,
	and an algorithm that receives as  input  an odd integer $r\geq 3,$ $\tilde{a},a,\ell\in \Bbb{N},$ where $\tilde{a}\leq a,$ a graph $G,$ a set $A\subseteq V(G)$ of size at most $a,$ and a flatness pair $(W,\frR)$ of $G\setminus A$ of height $\funref{label_distinguimos}(r,a,z,\ell),$
	and outputs
	a flatness pair $(\breve{W},\breve{\frR})$ of $G\setminus A$  of height $r$ that is $\ell$-homogeneous with respect to ${A \choose \leq {\tilde{a}}}$ and  is
	a $W'$-tilt of $(W,\frR)$ for some subwall $W'$ of $W.$
	Moreover, $\funref{label_distinguimos}(r,a,\tilde{a},\ell) = {\cal O}(r^{\funref{label_irresistibility}(a,\tilde{a},\ell)})$ and the algorithm runs in time {$2^{{\cal O}(\funref{label_irresistibility}(a,\tilde{a},\ell)\cdot r \log r)}\cdot(n+m)$}.
\end{proposition}


\paragraph{The price of homogeneity.}
As \autoref{label_highlighting} indicates, finding a homogeneous flat wall inside a flat wall has a price, corresponding to
the function $\funref{label_distinguimos}(r,a,\tilde{a},\ell)$ of the required height of the given flat wall.
The ``polynomial gap'' between the height of the given flatness pair $(W,\frak{R})$ of $G\setminus A$ and the homogenous flat wall that is returned is determined by the function $\funref{label_irresistibility}$ of \autoref{label_superintendent}, that bounds the number of different folios that
can be rooted through the augmented flaps of $(W,\frak{R}),$ for each possible augmentation of each flap with a subset of $A$ of size at most $\tilde{a}.$
In this paper, we will use \autoref{label_highlighting} in order to compute a flat wall that is homogeneous with respect to $2^A,$ that is, for $\tilde{a} = a.$
We set $c_{a,\ell}=\funref{label_irresistibility}(a,a,\ell)$ and we point out that, in general, it follows that $c_{a,\ell}=2^{2^{{\cal O}((a+\ell)\cdot \log(a+\ell))}}.$
However, by using the notion of representatives instead of folios as in~\cite{BasteST20acom},
we can obtain a smaller bound of $c_{a,\ell}=2^{{\cal O}((a+\ell)\cdot \log(a+\ell))}.$
We call $c_{a_{\cal F},\ell_{\cal F}}$ the {\em palette-variety} of  ${\cal F}.$

\section{Auxiliary algorithmic and combinatorial results}\label{label_compensating}

In this section we provide some algorithmic and combinatorial results that will support the main algorithms of this paper. In \autoref{label_desencajaron} we provide an algorithm that finds an irrelevant vertex inside a ``large enough'' homogeneous  flat wall, while, in \autoref{label_bezeichnungsweise}, we define canonical partitions of walls and we present some combinatorial results that allow our algorithms to branch.

\subsection{Finding an irrelevant vertex}
\label{label_desencajaron}

The {\sl irrelevant vertex technique} was introduced in~\cite{RobertsonS95XIII}
for providing an {\sf FPT}-algorithm for the \mnb{\sc Disjoint Paths} problem.
Moreover, this technique has appeared to be quite versatile and is now a standard tool
of parameterized algorithm design (see e.g.,~\cite{CyganFKLMPPS15para,ThilikosBDFM12grap}).
The applicability of  this technique for \mnb{\sc ${\cal F}$-M-Deletion} is materialized in this section by the algorithm of \autoref{label_appressavamo}.

For the proof of \autoref{label_appressavamo} we need the next combinatorial result, \autoref{label_perspicacity}, whose
proof is presented in \cite{SauST21kapiI}.
\autoref{label_perspicacity}
intuitively states that, given a graph $G$ and a homogeneous flatness pair $(W,\frak{R})$ of $G$ of ``big enough'' height, it holds that
for every $W^{(q)}$-tilt $(\hat{W},\hat{\frak{R}})$ of $(W,\frak{R}),$ ${\sf compass}_{\hat{\frak{R}}}(\hat{W})$
can be ``safely'' removed from the input graph $G,$ in the sense that $(G,k)$  and
$(G\setminus V({\sf compass}_{\hat{\frak{R}}}(\hat{W})),k)$ are equivalent instances of {\sc ${\cal F}$-M-Deletion}.
In \cite{SauST21kapiI} we insisted on a proof of \autoref{label_perspicacity}
that requires homogeneity with respect to $2^A$ (that is, with respect to {\sl every} possible subset of $A$) in order to find an irrelevant wall, no matter the choice of the hitting set.
This is an enhancement of the result of \cite[Theorem 5.2]{BasteST20acom},
which allows to reroute minors outside a part of the wall that is homogeneous with respect to a {\sl particular} apex set.
When aiming to detect a wall that is irrelevant for {\sc ${\cal F}$-M-Deletion},
we do not know a priori which is the hitting set, and therefore we need to ask, firstly, for a wall that is irrelevant for every choice of a hitting set $S$ and, secondly, for homogeneity that captures all possible remaining (after the deletion of $S$) apex sets of the flat wall in order to apply \cite[Theorem 5.2]{BasteST20acom} for such an apex set.
For these reasons, 
this enhancement of \cite[Theorem 5.2]{BasteST20acom} is essential for our case.

The running time of the next result depends on the function $f_{\sf ul}$ coming from the Unique Linkage Theorem from~\cite{KawarabayashiW10asho} (see also \cite{RobertsonS09XXI,RobertsonSGM22}).
Recall that $\ell_{\cal F}=\max\{{\sf detail}(H)\mid H\in{\cal F}\}.$

\begin{proposition}\label{label_perspicacity}
	There exist two functions $\newfun{label_objectivement}: \Bbb{N}^{4}\to \Bbb{N}$ and $\newfun{label_preoccupations}:\Bbb{N}^2\to \Bbb{N},$
	{where the images of $\funref{label_objectivement}$ are odd numbers,} such that for every $a,k\in\Bbb{N},$ every odd $q\in\Bbb{N}_{\geq 3},$ and every graph $G,$
	if $A$ is a subset of $V(G)$ of size at most $a$
	and
	$(W,\frR)$ is a regular flatness pair of $G\setminus A$ of height at least $\funref{label_objectivement}(a,\ell_{\cal F},q,k)$ that is $\funref{label_preoccupations}(a,\ell_{\cal F})$-homogeneous with respect to $2^A,$
	then for every  $W^{(q)}$-tilt $(\hat{W},\hat{\frak{R}})$ of $(W,\frak{R}),$
	it holds that   $(G,k)$ and $(G\setminus V({\sf compass}_{\hat{\frR}}(\hat{W})),k)$ are equivalent instances of \mnb{\sc ${\cal F}$-M-Deletion}.
	Moreover, $\funref{label_objectivement}(a,\ell_{\cal F},q,k)={\cal O}(k\cdot (f_{\sf ul}(16a+12\ell_{\cal F}))^3   + q ),$ where $f_{\sf ul}$ is the function of the Unique Linkage Theorem, and $\funref{label_preoccupations}(a,\ell_{\cal F})= a+\ell_{\cal F} +3.$\end{proposition}


By applying \autoref{label_protectively} on top of \autoref{label_perspicacity}, in order to find a tilt that is guaranteed to be irrelevant by \autoref{label_perspicacity},
we directly get the following algorithm,
which outputs a flatness pair $(\hat{W},\hat{\frR})$ of an input graph $G$ such that $(G,k)$  and $(G\setminus V({\sf compass}_{\hat{\frR}}(\hat{W})),k)$ are equivalent instances of {\sc ${\cal F}$-M-Deletion}. In fact, in the rest of the paper, we use a slightly weaker version of \autoref{label_appressavamo}, referred as \autoref{label_mitinbegriffen}, that outputs just an irrelevant vertex. Here,  we prove this more general result for future use.

\begin{lemma}\label{label_appressavamo}
	There exists an algorithm with the following specifications:\medskip

	\noindent{\tt Find-Irrelevant-Wall}$(k,q,a,G,A,W,\frR)$\\
	\noindent{\textbf{Input}:} Three integers $k,q, a\in\Bbb{N},$ with odd $q \geq 3,$ a graph $G,$ a set $A\subseteq V(G)$ of size at most $a,$
	and a regular flatness pair $(W,\frR)$ of $G\setminus A$ of height at least $\funref{label_objectivement}(a,\ell_{\cal F},q,k)$ that is  $\funref{label_preoccupations}(a,\ell_{\cal F})$-homogeneous with respect to $2^A.$\\
	\noindent{\textbf{Output}:} A flatness pair $(\hat{W},\hat{\frR})$ of $G\setminus A$ that is a $W^{(q)}$-tilt of $(W,\frR)$ and
	such that $(G,k)$ and $(G\setminus V({\sf compass}_{\hat{\frR}}(\hat{W})),k)$ are equivalent instances of \mnb{\sc ${\cal F}$-M-Deletion}.\\
	Moreover,
	this algorithm runs in ${\cal O}(n+m)$-time.
\end{lemma}

Notice that \autoref{label_appressavamo} together with \autoref{label_stepdaughter}
imply \autoref{label_mitinbegriffen} if we set $q=3$ and output a central vertex of the obtained $3$-wall.

\begin{corollary}\label{label_mitinbegriffen}
	There exists an algorithm with the following specifications:\medskip

	\noindent{\tt Find-Irrelevant-Vertex}$(k,a,G,A,W,\frR)$\\
	\noindent{\textbf{Input}:} Two integers $k,a\in\Bbb{N},$ a graph $G,$ a set $A\subseteq V(G)$ of size at most $a,$ and a regular flatness pair $(W,\frR)$ of $G\setminus A$ of height at least $\funref{label_objectivement}(a,\ell_{\cal F},3,k)$ that is  $\funref{label_preoccupations}(a,\ell_{\cal F})$-homogeneous with respect to $2^A.$\\
	\noindent{\textbf{Output:}} A vertex $v\in V(G)$ such that $(G,k)$ and $(G\setminus v,k)$ are equivalent instances of \mnb{\sc ${\cal F}$-M-Deletion}.\\
	Moreover, 
	this algorithm runs in ${\cal O}(n+m)$-time.
\end{corollary}

\subsection{Combinatorial results for branching}\label{label_bezeichnungsweise}
In this subsection, we present the notion of a canonical partition and provide two combinatorial results that will justify a branching step of our  algorithm and, if such a step cannot be applied, the existence of a wall that will allow the application of the irrelevant vertex technique.
Canonical partitions were introduced in~\cite{SauST21kapiI}.

\paragraph{Canonical partitions.}
Let $r\geq 3$ be an odd integer.
Let $W$ be an $r$-wall and let $P_{1}, \ldots, P_{r}$ (resp. $L_{1},\ldots, L_{r}$) be its vertical (resp. horizontal) paths.
For every even  (resp. odd) $i\in[2,r-1]$ and every $j\in[2,r-1],$ we define ${A}^{(i,j)}$ to be the  subpath of $P_{i}$ that starts from a vertex of $P_{i}\cap L_{j}$ and finishes at a neighbor of a vertex in $L_{j+1}$ (resp. $L_{j-1}$), such that $P_{i}\cap L_{j}\subseteq A^{(i,j)}$ and $A^{(i,j)}$ does not intersect $L_{j+1}$ (resp. $L_{j-1}$).
Similarly, for every $i,j\in[2,r-1],$ we define $B^{(i,j)}$ to be the subpath of $L_{j}$ that starts from a vertex of $P_{i}\cap L_{j}$ and finishes at a neighbor of a vertex in $P_{i-1},$ such that $P_{i}\cap L_{j}\subseteq A^{(i,j)}$ and $A^{(i,j)}$ does not intersect $P_{i-1}.$

\begin{figure}[h]
	\centering
	\includegraphics[width=6cm]{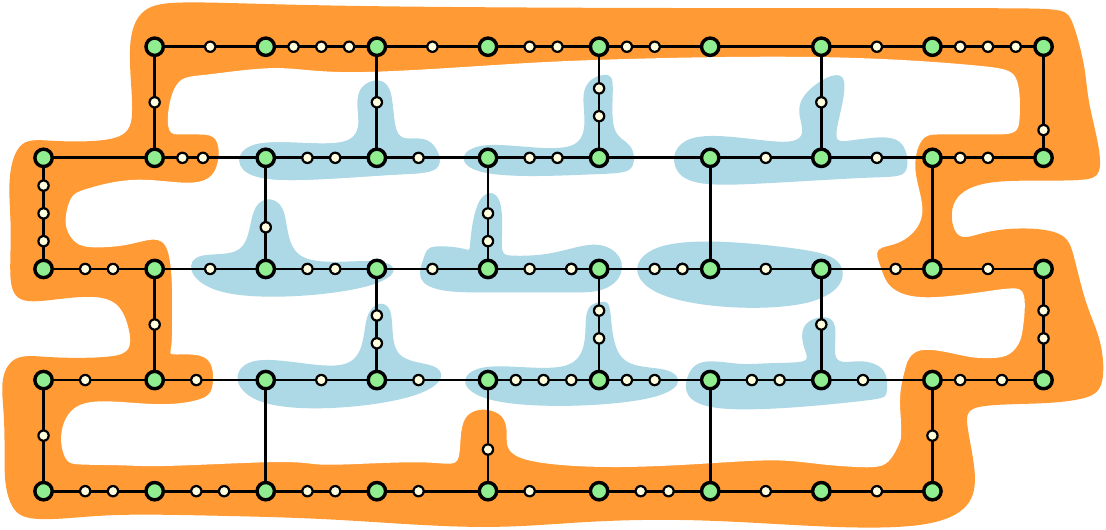}
	\caption{\small A $5$-wall and its canonical partition ${\cal Q}.$ The orange bag is the external bag $Q_{\rm ext}.$}
	\label{label_incidentally}
\end{figure}

For every  $i,j\in[2,r-1],$ we denote by $Q^{(i,j)}$ the graph $A^{(i,j)}\cup B^{(i,j)}$ and ${Q_{\rm ext}}$ to be the graph $W\setminus \bigcup_{i,j\in[2,r-1]} Q_{i,j}.$
Now consider the collection ${\cal Q}=\{Q_{\rm ext}\}\cup\{Q_{i,j}\mid i,j\in[2,r-1]\}$
and observe that the graphs in ${\cal Q}$ are connected subgraphs of $W$ and their vertex sets form a partition of $V(W).$
We call ${\cal Q}$ the {\em canonical partition} of $W.$ Also, we call every $Q_{i,j}, i,j\in[2,r-1]$ an {\em internal bag} of ${\cal Q},$ while we refer to $Q_{\rm ext}$ as the {\em external bag} of ${\cal Q}.$ See \autoref{label_incidentally} for an illustration of the notions defined above.
For every $i\in[(r-1)/2],$ we say that a set $Q\in {\cal Q}$ is an {\em $i$-internal bag of ${\cal Q}$} if $V(Q)$ does not contain any vertex of the first $i$ layers of $W.$
Notice that the $1$-internal bags of ${\cal Q}$ are the internal bags of ${\cal Q}.$

Let $(W,\frak{R})$ be a flatness pair of a graph $G.$
Consider the canonical partition ${\cal Q}$ of $W.$ We enhance the graphs of ${\cal Q}$
so to include in them all the vertices of $G$ by applying the following procedure. We set $\tilde{\cal Q}:={\cal Q}$
and, as long as there is a vertex  $x\in V({\sf compass}_{\frak{R}}(W))\setminus V(\cupall \tilde{\cal Q})$ 
that is adjacent to a vertex of a graph $Q\in \tilde{\cal Q},$  update $\tilde{\cal Q}:=\tilde{\cal Q}\setminus \{Q\}\cup \{\tilde{Q}\},$ where $\tilde{Q}={\sf compass}_{\frak{R}}(W)[\{x\}\cup V(Q)].$ Since ${\sf compass}_{\frR}(W)$ is a connected graph, in this way we define a partition of the vertices of ${\sf compass}_{\frak{R}}(W)$ into subsets inducing connected graphs.
We call the $\tilde{Q}\in\tilde{\cal Q}$ that contains $Q_{\rm ext}$ as a subgraph the {\em external bag} of $\tilde{\cal Q},$ and we denote it by $\tilde{Q}_{\rm ext},$ while we call {\em internal bags} of $\tilde{\cal Q}$ all graphs in $\tilde{\cal Q}\setminus \{\tilde{Q}_{\rm ext}\}.$
Moreover, we enhance $\tilde{\cal Q}$ by adding all vertices of $G\setminus V({\sf compass}_\frR (W)$ in its external bag, i.e., by updating $\tilde{Q}_{\rm ext}: = G[V(\tilde{Q}_{\rm ext})\cup V(G\setminus V({\sf compass}_\frR (W))].$
We call such a  partition $\tilde{\cal Q}$ a {\em $(W,\frR)$-canonical partition of $G.$}
Notice that a $(W,\frR)$-canonical partition of $G$ is  not unique, since the sets in ${\cal Q}$ can be ``expanded'' arbitrarily when introducing vertex $x$.

Let $(W,\frR)$ be a flatness pair of a graph $G$ of height $r,$ for some $r\geq 3$ and $\tilde{\cal Q}$ be a $(W,\frR)$-canonical partition of $G.$
For every $i\in[(r-1)/2],$ we say that a set $Q\in \tilde{\cal Q}$ is an {\em $i$-internal bag of $\tilde{\cal Q}$} if it contains an $i$-internal bag of ${\cal Q}$ as a subgraph.
\medskip

Next we identify a combinatorial structure that guarantees the existence of a set of {$q={\cal O}_{s_{\cal F}}(k)$}
vertices that intersects every solution $S$ of  \mnb{\sc  ${\cal F}$-M-Deletion} with input $(G,k).$
This will permit  branching on $q$ simpler instances of the form $(G',k-1).$
Recall that $a_{\cal F}$ is the minimum
apex number of a graph in ${\cal F}.$
The following result is proved in~\cite{SauST21kapiI}.

\begin{proposition}
	\label{label_disviluppato}
	There exist three functions $\newfun{label_guildenstern}, \newfun{label_unterschieden},\newfun{label_straightforward}: \Bbb{N}^{3}\to \Bbb{N},$
	such that if
	${\cal F}$ is a finite set of graphs,
	$G$ is a graph,
	$k\in\Bbb{N},$
	$A$ is a subset of $V(G),$
	$(W,\frR)$ is a flatness pair of $G\setminus A$ of height at least $\funref{label_guildenstern}(a_{\cal F},s_{\cal F},k),$
	$\tilde{\cal Q}$ is a $(W,\frR)$-canonical partition of $G\setminus A,$
	$A'$ is a subset of vertices of $A$ that are adjacent, in $G,$ to vertices of at least $\funref{label_unterschieden}(a_{\cal F},s_{\cal F},k)$ $\funref{label_straightforward}(a_{\cal F},s_{\cal F},k)$-internal bags of $\tilde{\cal Q},$ and $|A'|\geq a_{\cal F},$
	then for every  set $S\subseteq V(G)$ of size at most $k$ such that $G\setminus S \in {\bf exc}({\cal F})$
	it holds that  $S\cap A'\neq\emptyset.$
	Moreover, $\funref{label_guildenstern}(a,s,k)={\cal O}(2^a \cdot  s^{5/2} \cdot k^{5/2}),$
	$\funref{label_unterschieden}(a,s,k)={\cal O}(2^a \cdot s^3 \cdot k^3),$ and $\funref{label_straightforward}(a,s,k)={\cal O}((a^2 +k)\cdot s),$ where $a=a_{\cal F}$ and $s= s_{\cal F}.$
\end{proposition}

The next result is also proved in \cite{SauST21kapiI} and
intuitively states that, given a flatness pair $(W,\frR)$ of ``big enough'' height and a $(W,\frR)$-canonical partition $\tilde{\cal Q}$ of $G,$ we can find a ``packing'' of subwalls of $W$ that are inside some central part of $W$ and that the vertex set of every internal bag of $\tilde{\cal Q}$ intersects the vertices of the flaps in the influence of at most one of these walls.
We will use this result in the case where the set $A'$ of \autoref{label_disviluppato} is ``small'', i.e., there are only ``few'' vertices in $A$ that have ``big enough'' degree with respect to the central part of the canonical partition, and therefore  \autoref{label_disviluppato} cannot justify branching.
Following the latter condition and \autoref{label_stereotypical}, we will be able to find a flatness pair with ``few'' apices so as to build irrelevant vertex arguments inside its compass.

\begin{proposition}\label{label_stereotypical}
	There exists a function $\newfun{label_internalization}: \Bbb{N}^3 \to \Bbb{N}$ such that if $p,z\in\Bbb{N}_{\geq 1},$ $x\in\Bbb{N}_{\geq 3}$ is an odd integer, $G$ is a graph, $(W,\frR)$ is a flatness pair of $G$ of height at least $\funref{label_internalization}(z,x,p),$ and $\tilde{\cal Q}$ is a $(W,\frR)$-canonical partition of $G,$ then
	there is a collection ${\cal W}=\{W_1, \ldots, W_z\}$ of $x$-subwalls of $W$ such that
	\begin{itemize}
		\item for every $i \in [z],$ $\cupall{\sf influence}_{\frR}(W_i)$ is a subgraph of $\cupall \{Q\mid Q \text{ is a $p$-internal bag of }\tilde{\cal Q}\}$ and

		\item for every $i,j\in[z],$ with $i\neq j,$ there is no internal bag of $\tilde{\cal Q}$ that contains vertices of both $V(\cupall{\sf influence}_\frR (W_i))$ and $V(\cupall{\sf influence}_\frR (W_j)).$
	\end{itemize}
	Moreover, $\funref{label_internalization}(z,x,p)= {\cal O}(\sqrt{z}\cdot x + p).$
\end{proposition}


\section{The general  algorithm}
\label{label_pizzighettona}
In this section we present the general algorithm for  \mnb{\sc  ${\cal F}$-M-Deletion}.
The existence of this algorithm proves \autoref{label_preliminaries}.
In \autoref{label_physiologically}, we explain how to employ the {\sl iterative compression technique} so as to ask for an algorithm for a new, more convenient to solve, problem and, in \autoref{label_particularidad}, we develop an algorithm for this new problem.

\subsection{Iterative compression}\label{label_physiologically}

In order to prove \autoref{label_preliminaries}, we apply the  iterative compression technique (introduced in \cite{ReedSV04find};
see also~\cite{CyganFKLMPPS15para}) and we  give a $2^{{\sf poly}(k)} \cdot n^2$-time algorithm for the following problem.

\begin{center}
	\fbox{
		\begin{minipage}{14.5cm}
			\noindent\mnb{\sc  ${\cal F}$-M-Deletion-Compression}\\
			\noindent\textbf{Input:}~~A graph $G,$ a $k\in\Bbb{N},$ and a set $S$ of size $k+1$ such that $G\setminus S\in {\bf exc}({\cal F}).$\\
			\textbf{Objective:}~~Find, if exists, a set $S'\subseteq V(G)$ of size at most $k$ such that $G\setminus S'\in {\bf exc}({\cal F}).$
		\end{minipage}
	}
\end{center}

In other words, given an input $(G,k,S)$ of \mnb{\sc  ${\cal F}$-M-Deletion-Compression}, we have at hand a graph $G$ and a ``slightly larger than $k$'' hitting set $S$, and we aim to find a hitting set of size at most $k,$ that is a certificate that $(G,k)$ is a \yes-instance of \mnb{\sc ${\cal F}$-M-Deletion}.
Given this set $S,$ we can directly assume that $G\setminus S$ does not contain a big clique as a minor and therefore we can deal with this minor-free graph, and thus, due to \autoref{label_dishonorable}, we can obtain either a tree decomposition of $G$ of ``small'' width (and solve the problem using the dynamic programming algorithm of \cite{BasteST20acom}), or a flat wall on top of which we build our branching and irrelevant vertex technique arguments.
In this way, we manage to avoid the ``big clique'' possible output of \autoref{label_dishonorable}.
However, this swifting from  \mnb{\sc  ${\cal F}$-M-Deletion} to \mnb{\sc  ${\cal F}$-M-Deletion-Compression} comes together with an extra linear factor in the running time of the algorithm, as observed in the following (see~\cite{CyganFKLMPPS15para}).

\begin{observation}
	\label{label_correspondas}
	If there is an algorithm solving \mnb{\sc  ${\cal F}$-M-Deletion-Compression} in  $f(k)\cdot n^{c}$-time, then there exists an algorithm solving \mnb{\sc  ${\cal F}$-M-Deletion} in  ${\cal O}(f(k)\cdot n^{c+1})$-time.
\end{observation}

In \autoref{label_particularidad} we prove that \mnb{\sc  ${\cal F}$-M-Deletion-Compression} can be solved in $2^{{\sf poly}(k)} \cdot n^2$-time  (\autoref{label_countenances}). This
along with \autoref{label_correspondas} yield \autoref{label_preliminaries}.

\subsection{The algorithm}\label{label_particularidad}
In this subsection we present the algorithm solving \mnb{\sc  ${\cal F}$-M-Deletion-Compression}.

We set  $\tilde{c}_{a,\ell}
	:=\funref{label_nompareilles}(a,\funref{label_preoccupations}(a,\ell)) = 2^{2^{{\cal O}((a+\ell)\cdot\log(a+\ell))}},$ where $\funref{label_nompareilles}$ is the number of different folios given in \autoref{label_surprendront} and $\funref{label_preoccupations}$ is the function given in \autoref{label_perspicacity}, in order to find an irrelevant vertex.

\begin{lemma}
	\label{label_countenances}
	Let ${\cal F}$ be a finite collection of graphs.
	There is an algorithm solving  \mnb{\sc  ${\cal F}$-M-Deletion-Compression}  in $2^{{\cal O}_s (k^{2\cdot ({{c} + 2})})}\cdot n^2$-time, where $a=a_{\cal F},$ $s = s_{\cal F},$ $\ell = \ell_{\cal F}$, and ${c}=\tilde{c}_{a, \ell}.$
\end{lemma}

\begin{proof}
	For simplicity, in this proof,  we use $c$ instead of $\tilde{c}_{a_{\cal F}, \ell_{\cal F}},$
	$s$ instead of $s_{\cal F},$ $\ell$ instead of $\ell_{\cal F},$ $a$ instead of $a_{\cal F},$ and recall that $\ell={\cal O}(s^2)$ and ${a\leq s}.$
	Also, we set
	\begin{align*}
		z=         & \ \funref{label_objectivement}(a-1,\ell, 3, k),                     &
		d =        & \ \funref{label_preoccupations}(a,\ell),                               &
		b=         & \ \funref{label_distinguimos}(z,a,a,d)={\cal O}_{s_{\cal F}} (k^{c}),   \\
		m=         & \ \funref{label_guildenstern}(a,s,k),                                &
		x=         & \ \funref{label_unterschieden}(a,s,k),                                 &
		l =        & \  (\funref{label_schematization}(s)+k+1)\cdot x,                         \\
		p=         & \ \funref{label_straightforward}(a,s,k),                                 &
		h=         & \ \funref{label_internalization}(l+1,b,p),                              &
		\mbox{and} & \ r=  {\sf odd}(\max\{m,h\})= {\cal O}_s (k^{ c + 2}).
	\end{align*}
	We present the algorithm  {\tt Solve-Compression}, whose input is  a quadruple $(G,k',k,S)$
	where $G$ is a graph, $k'$ and $k$ are  non-negative integers with $k'\leq k,$
	and $S$ is a subset of $V(G)$ such that $|S|= k$ and $G\setminus S\in {\bf exc}({\cal F}).$
	The algorithm returns, if it exists, a solution for \mnb{\sc  ${\cal F}$-M-Deletion} on $(G,k').$
	Certainly,  we may assume that $k'<k,$ otherwise $S$ is already a solution and we are done.
	The steps of the algorithm are the following:

	\paragraph{Step 1.}
	Run the algorithm of \autoref{label_dishonorable} with input $(G\setminus S, r, s).$
	Since $G\setminus S\in{\bf exc}({\cal F})$ and  ${\cal F}\leq_{\sf m} K_{s},$ the algorithm outputs, {in
	time $2^{{\cal O}_s (r^{2})}\cdot n=2^{{\cal O}_s (k^{2\cdot ({c + 2})})}\cdot n,$}
	either a tree decomposition of $G\setminus S$
	of width at most at most $\funref{label_inconsummabile}(s)\cdot r,$
	or
	a set $A\subseteq V(G)$  with $|A|\leq \funref{label_schematization}(s)$ and a regular flatness pair $(W,\frak{R})$ of $G\setminus A$ of height $r.$
	In the first case, we solve \mnb{\sc  ${\cal F}$-M-Deletion-Compression}
	{in time $2^{{\cal O}_s (r \log r)}\cdot n = 2^{{\cal O}_s (k^{ c + 2}\log k)}\cdot n$}
	using the algorithm of \autoref{label_confiscating}.
	In what follows we examine the second case, where the algorithm of \autoref{label_dishonorable} outputs
	a set $A\subseteq V(G)$  with $|A|\leq \funref{label_schematization}(s)$ and a regular flatness pair $(W,\frak{R})$ of $G\setminus A$ of height $r.$

	We consider a $(W,\frR)$-canonical partition $\tilde{\cal Q}$ of $G\setminus (S\cup A).$
	We compute, in ${\cal O}(n)$-time, the set $$A^\star = \{v\in S\cup A \mid  v\text{ is adjacent, in }G,\text{ to vertices of at least $x$ $p$-internal bags of }\tilde{\cal Q}\}$$ and we proceed to the second step.

	\paragraph{Step 2.} The algorithm examines two cases depending on the size of the $A^\star.$ In the first case, the {\sl branching case}, the outcome  is a set of vertices, the set $S\cup A,$ that should intersect every possible solution. In the second case, the {\sl irrelevant vertex case}, the outcome is an irrelevant vertex.
	\bigskip

	\noindent[{\sl Branching case}]. It holds that $|A^\star|\geq a.$
	{In this case the algorithm recursively calls {\tt Solve-Compression}
			with input $(G\setminus x,k'-1,|S\setminus x|,S\setminus x)$ for every
			$x\in A^\star,$ and if one of these new instances is a \yes-instance, certified by a set $\bar{S},$
			then returns $\bar{S}\cup\{x\},$ otherwise it returns that $(G,k')$ is a \no-instance.}\medskip

	\autoref{label_disviluppato} implies that the above branching step of the algorithm is correct.

	\bigskip

	\noindent[{\sl Irrelevant vertex case}]. It holds that $|A^\star|< a.$
	%
	We consider a family ${\cal W}=\{{W}_{1}, \ldots, {W}_{l+1}\}$ of $l+1$ $b$-subwalls of $W$ such that for every $i \in [l+1],$ $\cupall{\sf influence}_{\frR}(W_i)$ is a subgraph of $\cupall \{Q\mid Q \text{ is a $p$-internal bag of }\tilde{\cal Q}\}$ and for every $i,j\in[l+1],$ where $i\neq j,$ there is no internal bag $Q\in \tilde{\cal Q}$
	that contains vertices of both $V(\cupall{\sf influence}_\frR (W_i))$ and
	$V(\cupall{\sf influence}_\frR (W_j)).$
	The existence of ${\cal W}$ follows from the fact that $r\geq h = \funref{label_internalization}(l+1,b,p)$ and \autoref{label_stereotypical}.

	Notice that
	the vertices in $(S\cup A)\setminus A^\star$
	are adjacent, in $G,$ to vertices of  at most
	$x\cdot |(S\cup A)\setminus A^\star|\leq x\cdot (\funref{label_schematization}(s)+k+1)= l$
	$p$-internal bags of $\tilde{\cal Q}.$
	Hence, taking into account the aforementioned properties of the walls $W_1,\ldots, W_{l+1},$
	there exists an $i\in [l+1]$ such that no
	vertex in $(S\cup A)\setminus A^\star$ is
	adjacent to vertices of $V(\cupall{\sf influence}_{\frak{R}}({W}_{i})).$
	In other words, if there exists a vertex $v\in V(\cupall{\sf influence}_{\frak{R}}({W}_{i}))$
	that is adjacent, in $G,$ to a vertex $u\in S\cup A,$ then $u\in A^\star.$
	The fact that $|A^{\star}|< a$ implies that, for this $i,$ there are less than $a$ vertices in $S\cup A$
	that are adjacent to vertices of $V(\cupall{\sf influence}_{\frak{R}}({W}_{i})).$

	Since ${W}_{i}$ is a $b$-subwall of $W$ and $(W,\frR)$ is a flatness pair of $G\setminus (S\cup A),$
	we apply the algorithm of \autoref{label_protectively}, and obtain, in linear time,
	a flatness pair $(\tilde{W}_{i},\tilde{\frR}_{i})$ of $G\setminus (S\cup A)$ that is a ${W}_i$-tilt of $(W,\frR).$
	Notice that since $(\tilde{W}_{i},\tilde{\frR}_{i})$ is a ${W}_i$-tilt of $(W,\frR),$
	${\sf compass}_{\tilde{\frR}_{i}}(\tilde{W}_{i})$ is a subgraph of $\cupall{\sf influence}_{{\frak{R}}}(\bar{W}_i)$ and, due to \autoref{label_ressemblances}, $(\tilde{W}_{i},\tilde{\frR}_{i})$ is regular.
	This	implies that, if $A_i$ is the set of vertices from $S\cup A$ that
	are adjacent to vertices of ${\sf compass}_{\tilde{\frR}_{i}}(\tilde{W}_{i})$ in $G,$ then $A_i\subseteq A^\star$ and therefore $|A_i| < a.$
	Notice that, by adding the vertices of $(S\cup A)\setminus A_i$ to $G\setminus (S\cup A),$
	we obtain a flatness pair $(\tilde{W}_i, \tilde{\frR}_i ')$ of $G\setminus A_i$ such that ${\sf compass}_{\tilde{\frR}_{i}}(\tilde{W}_{i})={\sf compass}_{\tilde{\frR}_i '}(\tilde{W}_i).$
	Applying the algorithm of \autoref{label_highlighting} for $(b,a,a,d, G,A_i, \tilde{W}_i,\tilde{\frR}_i '),$ we obtain a flatness pair  $(\breve{W}_{i},\breve{\frR}_{i})$  of $G\setminus A_i$ of height $z$ that is  $d$-homogeneous with respect to $2^{A_i}$ and is a $\tilde{W}_i '$-tilt of $(\tilde{W}_i,\tilde{\frR}_i)$ for some subwall $\tilde{W}_i '$ of $\tilde{W}_i.$ Due to \autoref{label_ressemblances},  $(\breve{W}_{i},\breve{\frR}_{i})$ is also regular.
	This algorithm runs in $2^{{\cal O}_s (k\log k)}\cdot n$-time.

	We now apply {\tt Find-Irrelevant-Vertex} of \autoref{label_mitinbegriffen}  for  $(k,a,G,A_i,\breve{W}_{i},\breve{\frR}_{i})$
	and obtain a vertex $v$
	such that $(G,k)$ and $(G\setminus v,k)$
	are equivalent instances of \mnb{\sc  ${\cal F}$-M-Deletion}.
	According to \autoref{label_mitinbegriffen}, this vertex can be detected in linear time,
	and {the algorithm correctly  calls recursively {\tt Solve-Compression}
			with input $(G\setminus v,k',k,S).$} This completes the {\sl irrelevant vertex case}.
	\medskip

	Recall that $|S\cup A|\leq k+1+ \funref{label_schematization}(s)={\cal O}_{s}(k).$ Therefore,
	if $T(n,k',k)$ is the running time of the above algorithm, then
	$$T(n,k',k)\ \leq\ 2^{{\cal O}_s (k^{2\cdot ({ c + 2})})}\cdot n+\max\{T(n-1,k',k),{\cal O}_{s}(k)\cdot  T(n,k-1,k)\}$$
	that, given that $k'\leq k,$ implies that $T(n,k',k)=2^{{\cal O}_s (k^{2\cdot ({c + 2})})}\cdot n^2.$

	Notice now that the output of {\tt Solve-Compression} on $(G,k,k+1,S)$
	gives a solution for \mnb{\sc  ${\cal F}$-M-Deletion-Compression}
	on this instance.
\end{proof}

\section{The apex-minor free case}
\label{label_demonstrieren}


In this section we present an improved algorithm solving \mnb{\sc  ${\cal F}$-M-Deletion}  in the case where $a_{\cal F}=1.$
The existence of this algorithm proves \autoref{label_entrepreneur}.
In \autoref{label_condamnaient}, we show that a graph that contains a flat wall that is ``highly connected'' to a vertex in its apex set, also contains any apex graph as a minor.
In \autoref{label_unquestioning}, we provide an algorithm that will allow us to detect a wall inside a graph $G$ in linear time.
In \autoref{label_desvanecerse}, we provide the improved algorithm that solves \mnb{\sc  ${\cal F}$-M-Deletion}  in the case where $a_{\cal F}=1$ and, in \autoref{label_sleeplessness}, we prove its correctness.

\subsection{Finding an apex graph as a minor}
\label{label_condamnaient}

\paragraph{Grids.}
Let $k\in\Bbb{N}_{\geq 2}.$
We use the term $k$-{\em grid} to refer to the $(k\times k)$-grid. We say that a graph is a {\em partially triangulated $r$-grid} if it can be obtained from an $r$-grid after adding edges in such a way that the remaining graph remains planar.

Let  $k,r\in\Bbb{N}_{\geq 2}.$
A vertex of a $(k\times r)$-grid is called
	{\em internal} if it has degree four, and otherwise it is called {\em external}.
We define the {\em perimeter} of a $(k\times r)$-grid to be the unique cycle of the grid of length at least three that does not contain internal vertices.

\begin{figure}[h]
	\centering
	\includegraphics[width=3.5cm]{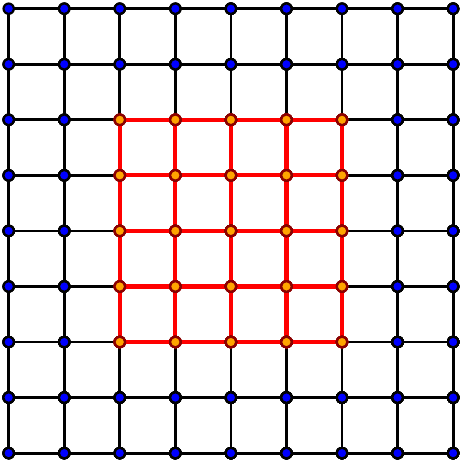}
	\caption{A 9-grid and its central 5-grid.}
	\label{label_unconciliatory}
\end{figure}

Let $r\in \Bbb{N}_{\geq 2}$ and $H$ be an $r$-grid.
Given an $i\in\lceil \frac{r}{2}\rceil,$ we define the {\em $i$-th layer} of $H$ recursively as follows.
The first layer of $H$ is its perimeter, while, if $i\geq 2,$ the $i$-th layer of $H$ is
the $(i-1)$-th layer of the grid created if we remove from $H$ its perimeter.
Given two  odd integers $r,q\in\Bbb{N}_{\geq 3}$ such that $q\leq r$ and an $r$-grid $H,$
we define the {\em central $q$-grid} of $H$ to be the graph obtained from $H$
if we remove from $H$ its $\frac{r-q}{2}$ first layers.
See \autoref{label_unconciliatory} for an illustration of the notions defined above.
Given a partially triangulated $r$-grid $H,$ we call {\em central $q$-grid} of $H$ the subgraph of $H$ induced by the vertices of the central $q$-grid of the underlying grid of $H.$

Given a graph $G$ and a vertex $v\in V(G),$ we say
that a graph $H$ is a {\em $v$-fixed contraction} of $G$ if $H$ can be obtained from $G$ after contracting edges that are not incident with $v.$
A graph $H$ is a {\em $v$-apex partially triangulated $r$-grid} if  it can be obtained from a partially triangulated $r$-grid $\Gamma$ after adding a new vertex $v$ and some edges between $v$ and vertices in $V(\Gamma).$
Α {\em complete $v$-apex partially triangulated $r$-grid} is a graph obtained from a $v$-apex partially triangulated $r$-grid by adding every edge between $v$ and the vertices of the grid.

The following result is a special case of~\cite[Lemma 29]{SauST21amor}.
\begin{proposition}\label{label_dvflptoffous}
	There exist three functions $\newfun{label_functionalist}, \newfun{label_entreteneros},$ and $\newfun{label_improvisamente}: \Bbb{N}\to \Bbb{N},$ such that if $r\in \Bbb{N},$
	$H$ is a $v$-apex partially triangulated $h$-grid, where $v\in V(H)$ and $h\geq \funref{label_functionalist}(r)+2\cdot \funref{label_improvisamente}(r),$
	and
	vertex $v$ has at least $\funref{label_entreteneros}(r)$ neighbors in the central
	$\funref{label_functionalist}(r)$-grid of $H\setminus \{v\},$
	then $H$ contains as a $v$-fixed contraction a complete $v$-apex partially triangulated $r$-grid.
	Moreover,  $\funref{label_functionalist}(r)= {\cal O}(r^5),$
	$\funref{label_entreteneros}(r)={\cal O}(r^{6}),$ and $\funref{label_improvisamente}(r)={\cal O}(r^{2}).$
\end{proposition}

The following easy observation intuitively states that every planar graph $H$ is a minor of a big enough grid, where the relationship between the size of the grid and $|V(H)|$ is linear (see e.g.,~\cite{RobertsonST94quic}).
\begin{proposition}\label{label_confederates}
	There is a function $\newfun{label_connotacions}:\Bbb{N}\to\Bbb{N}$ such that every planar graph on $n$ vertices is a minor of the
	$\funref{label_connotacions}(n)$-grid. Moreover, $\funref{label_connotacions}(n)={\cal O}(n).$
\end{proposition}

In the proof of  \autoref{label_entrepreneur}, we will need the following result.

\begin{lemma}
	\label{label_persuadieran}
	There exist three functions
	$\newfun{label_biancheggiava}, \newfun{label_disbnguished}, \newfun{label_accroissement}: \Bbb{N}\to \Bbb{N},$
	such that if  ${\cal F}$ is a finite set of graphs containing an apex graph,
	$G$ is a graph, $A$ is a subset of $V(G),$
	$(W,\frR)$ is a flatness pair of $G\setminus A$ of height at least $\funref{label_biancheggiava}(s_{\cal F}),$
	$\tilde{\cal Q}$ is a $(W,\frR)$-canonical partition of $G\setminus A,$ and there is a vertex in $A$ that is adjacent, in $G,$ to at least $\funref{label_disbnguished}(s_{\cal F})$
	$\funref{label_accroissement}(s_{\cal F})$-internal bags of $\tilde{\cal Q},$ then ${\cal F}\leq_{\sf m} G.$
\end{lemma}

\begin{proof}
	Let $\funref{label_functionalist}, \funref{label_entreteneros},$ and $\funref{label_improvisamente}$ be the functions of \autoref{label_dvflptoffous} and $\funref{label_connotacions}$ be the function of \autoref{label_confederates}.
	We set $r = \funref{label_connotacions}(s_{\cal F}-1),$
	$\funref{label_biancheggiava}(s_{\cal F}) = \funref{label_functionalist}(r) + 2\cdot \funref{label_improvisamente}(r) + 2,$
	$\funref{label_disbnguished}(s_{\cal F}) = \funref{label_entreteneros}(r),$ and
	$\funref{label_accroissement}(s_{\cal F}) = \funref{label_improvisamente}(r).$
	Let $G$ be a graph, $A\subseteq V(G),$ $(W,\frR)$ be a flatness pair of $G\setminus A$ of height $h,$ where $h\geq \funref{label_biancheggiava}(s_{\cal F}),$ $\tilde{\cal Q}$ be a $(W,\frR)$-canonical partition of $G\setminus A,$ and $v$ be a vertex in $A$ that is adjacent, in $G,$ to at least $\funref{label_disbnguished}(s_{\cal F})$
	$\funref{label_accroissement}(s_{\cal F})$-internal bags of $\tilde{\cal Q}.$

	We contract every bag in $\tilde{\cal Q}$ to a vertex.
	Observe that this results in a planar graph (since $(W,\frR)$ is a flatness pair) that is a partially triangulated $(h-2)$-grid $\bar{\Gamma}$ (whose vertices correspond to the internal bags of $\tilde{\cal Q}$) together with an extra vertex $u_{\sf ext}$ (which corresponds to the external bag of $\tilde{\cal Q}$) that is adjacent to all the vertices in the perimeter of $\bar{\Gamma}.$
	We contract an edge between $u_{\sf ext}$ and a vertex in the perimeter of $\bar{\Gamma}$
	and we denote by $\Gamma$ the obtained partially triangulated $(h-2)$-grid.
	We set $\Gamma^{+v}$ to be the graph obtained from $\Gamma$ by adding the vertex $v$ and the edges $\{v,u\},$ if $u$ is a vertex of $\Gamma$ that corresponds to a bag $Q\in \tilde{\cal Q}$ that contains a vertex adjacent, in $G$, to $v$.
	Notice that $\Gamma^{+v}$ is a  $v$-apex partially triangulated $(h-2)$-grid that is a minor of $G.$
	Moreover, observe that since $v$ is adjacent, in $G,$ to vertices of an $\funref{label_accroissement}(s_{\cal F})$-internal bag of $\tilde{\cal Q},$
	then, since $\funref{label_accroissement}(s_{\cal F})=\funref{label_improvisamente}(r)$ and $h-2\geq  \funref{label_functionalist}(r) + 2\cdot \funref{label_improvisamente}(r),$ vertex $v$ is also adjacent to a vertex in the central $\funref{label_accroissement}(s_{\cal F})$-grid of $\Gamma = \Gamma^{+v}\setminus \{v\}.$
	Thus, $v$ has at least $\funref{label_disbnguished}(s_{\cal F})$ neighbors  in the central  $\funref{label_accroissement}(s_{\cal F})$-grid of $\Gamma.$
	By \autoref{label_dvflptoffous}, $\Gamma^{+v}$ contains as a $v$-fixed contraction a complete $v$-apex $r$-grid and therefore, since $r = \funref{label_connotacions}(s_{\cal F}-1),$ by \autoref{label_confederates} every apex graph on at most $s_{\cal F}$ vertices is a minor of $G.$
	Thus, ${\cal F}\leq_{\sf m} G,$ and the lemma follows.
\end{proof}

\subsection{Quickly finding a wall}\label{label_unquestioning}
In this subsection we prove \autoref{label_improvements} that intuitively states that there is an algorithm that,
given a graph $G$ and two non-negative integers $r$ and $k,$  outputs either that
$(G,k)$ is a \no-instance of \mnb{\sc  ${\cal F}$-M-Deletion},
or a report that the treewidth of $G$ is polynomially bounded by $r$ and $k,$ or an $r$-wall of $G.$
Before stating \autoref{label_improvements}, we present the following result
of Kawarabayashi and Kobayashi \cite{KawarabayashiK20line},
which provides a {\sl linear} relation between the treewidth and the height of a largest wall in a minor-free graph.

\begin{proposition}\label{label_bureaucracies}
	There is a function $\newfun{label_accouchement}:\Bbb{N}\to \Bbb{N}$ such that, for every $t,r\in \Bbb{N}$
	and every graph $G$ that does not contain $K_{t}$ as a minor, if $\tw(G)\geq \funref{label_accouchement}(t)\cdot r,$ then $G$ contains an $r$-wall as a subgraph.
	In particular, one may choose $\funref{label_accouchement}(t)=2^{{\cal O}(t^{2} \cdot \log t)}.$
\end{proposition}

\autoref{label_improvements} is a variation of \cite[Lemma 11]{SauST21amor} that we prove in this subsection.
The version presented here will be useful for the design of the algorithm of \autoref{label_entrepreneur}.
Recall that  $s_{\cal F}=\max\{|V(H)|\mid H\in {\cal F}\}.$

\begin{lemma}\label{label_improvements}
	There exists an algorithm with the following specifications:\\

	\noindent{\tt Find-Wall}$(G,r,k)$\\
	\noindent{\textbf{Input}:} A graph $G,$ an odd $r\in\Bbb{N}_{\geq 3},$ and a $k\in\Bbb{N}.$\\
	\noindent{\textbf{Output}:} One of the following:
	\begin{itemize}
		\item Either a report that $G$ has treewidth at most $\funref{label_accouchement}(s_{\cal F})\cdot r+k,$ or
		\item an $r$-wall $W$ of $G,$ or
		\item a report that $(G,k)$ is a \no-instance of \mnb{\sc  ${\cal F}$-M-Deletion}.
	\end{itemize}
	Moreover, this algorithm runs in  $2^{{\cal O}_{s_{\cal F}}(r^2+(k+r) \cdot \log (k+r))}\cdot n$-time.
\end{lemma}

The algorithm of  \autoref{label_improvements} is a recursive one. Namely, given an instance of this algorithm, we compute a smaller-sized instance and recurse.
This is achieved by using the following result that is derived from \cite{PerkovicR00anim}.
For a detailed analysis of the results of \cite{PerkovicR00anim}, see \cite{AlthausZ19opti}.

\begin{proposition}\label{label_preguntarnos}
	There exists an algorithm with the following specifications:\medskip

	\noindent{\textbf{Input}:}	A graph $G$ and a $t\in\Bbb{N}$ such that $|V(G)|\geq 12t^{3}.$\\
	\noindent{\textbf{Output}:} A graph $G^{\star}$ such that $|V(G^{\star})|\leq (1-\frac{1}{16t^{2}}) \cdot |V(G)|$ and:
	\begin{itemize}
		\item Either $G^{\star}$ is a subgraph of $G$ such that $\tw(G)=\tw({G^\star}),$  or
		\item $G^{\star}$ is obtained from $G$ after contracting the edges of a matching in $G$.
	\end{itemize}
	Moreover, this algorithm runs in $2^{{\cal O}(t)} \cdot n$-time.
\end{proposition}

We now have all the ingredients to prove \autoref{label_improvements}.

\begin{proof}[Proof of \autoref{label_improvements}]
	We set $c:=\funref{label_accouchement}(s_{\cal F})\cdot r+k.$
	We now describe a recursive algorithm as follows.\medskip

	We first argue for the base case, namely when $|V(G)| < 12c^{3}.$
	To check whether ${\tw}(G)\leq c,$ we use the algorithm of Arnborg et al. \cite{ArnborgCP87comp}, which runs in time ${\cal O}(|V(G)|^{c+2})=2^{{\cal O}_{s_{\cal F}} ((r+k)\cdot \log (r+k))},$ and if this is the case, we report the same and stop.
	If not, we aim to find an $r$-wall of $G$ or conclude that we are dealing with a {\sf no}-instance.
	We first consider an arbitrary ordering
	$(v_1, \dots, v_{|V(G)|})$ of the vertices of $G.$
	For each $i\in[|V(G)|],$ we set $G_i$ to be the graph induced by the vertices $v_1, \dots, v_i.$
	We iteratively run the algorithm of \autoref{label_panathinaiko} on $G_i$ and $c$ for increasing values of $i.$
	This algorithm runs in $2^{{\cal O}(c)}\cdot |V(G)|=2^{{\cal O}_{s_{\cal F}} (r+k)}$-time.
	Let $j\in[|V(G)|]$ be the smallest integer such that the above algorithm outputs a report that ${\tw}(G_j)>c$
	and notice that there exists a tree decomposition $({\cal T}_j, \chi_j)$ of $G_j$ (obtained by the one of $G_{j-1}$ by adding the vertex $v_j$ to all the bags) of width at most $5c+5.$
	Thus, we can call the algorithm of \autoref{label_confiscating} with input $(G_j, 5c+5,k)$ (which runs in $2^{{\cal O}_{s_{\cal F}}(c \cdot \log c)}\cdot |V(G_j)|=2^{{\cal O}_{s_{\cal F}}((r+k)\cdot \log (r+k))}$-time) in order to find, if it exists, a set $S_j \subseteq V(G_j)$ such that $|S_j |\leq k$ and ${\cal F}\nleq_{{\sf m}} G_j\setminus S_j.$ We distinguish two cases.

	\begin{itemize}
		\item[$\bullet$] If such a set $S_j$ does not exist, then we can safely report that $(G,k)$ is a \no-instance.
		\item[$\bullet$] If such a set $S_j$ exists, then we call the algorithm of \autoref{label_sleepwalkers} for $G_j\setminus S_j$ (and the decomposition of $G_j\setminus S_j$ obtained from  $({\cal T}_j, \chi_j)$ by removing the vertices of $S_j$ from the all the bags in order to check whether it contains an elementary $r$-wall $W$ as a minor.
		      This algorithm runs in $2^{{\cal O}(c \cdot \log c)}\cdot r^{{\cal O}(c)}\cdot 2^{{\cal O}(r^{2})}\cdot |V(G_j\setminus S_j)|=2^{{\cal O}_{s_{\cal F}}((r+k) \cdot \log (r+k))}\cdot r^{{\cal O}_{s_{\cal F}}(r+k)}\cdot 2^{{\cal O}(r^2)}=2^{{\cal O}_{s_{\cal F}}(r^2+(r+k) \cdot \log (r+k))}$-time, since $|E(W)|={\cal O}(r^2).$
		      Since  $G_j\setminus S_j$ does not contain $K_{s_{\cal F}}$ as a minor and  $\tw(G_j\setminus S_j)\geq c-k=\funref{label_accouchement}(s_{\cal F})\cdot r$ and because of  \autoref{label_bureaucracies}, this algorithm will output an elementary $r$-wall $W$ of $G_j \setminus S_j.$
		      We also return $W$ as a wall of $G.$
	\end{itemize}
	Therefore, in the case where $|V(G)|< 12c^3$, we obtain one of the three possible outputs in time $2^{{\cal O}_{s_{\cal F}}(r^2+(r+k)\log (r+k))}.$
	\medskip

	If  $|V(G)|\geq 12c^{3},$ then we call  the algorithm of \autoref{label_preguntarnos} with input $(G,c),$ which outputs a graph $G^{\star}$ such that $|V(G^{\star})|\leq (1-\frac{1}{16c^{2}}) \cdot |V(G)|$ and
	\begin{itemize}
		\item either $G^{\star}$ is a subgraph of $G$ such that $\tw(G)=\tw(G^{\star}),$ or
		\item $G^{\star}$ is obtained from $G$  after contracting the edges of a matching in $G$.
	\end{itemize}


	In both cases, we recursively call the algorithm on  $G^{\star}$ and we distinguish the following two cases.\medskip

	\noindent{\em Case 1}: $G^{\star}$ is a subgraph of $G$ such that $\tw(G)=\tw(G^{\star}).$

	\begin{itemize}
		\item[(a)] 	If the recursive call on $G^\star$ reports that $\tw(G^{\star})\leq c,$ then we return that $\tw(G)\leq c.$

		\item[(b)] 	If the recursive call on $G^\star$ outputs an $r$-wall $W$ of $G^{\star},$ then we return $W$ as a wall of $G.$

		\item[(c)] 	If $(G^{\star},k)$ is a \no-instance, then we report that $(G,k)$ is also a {\sf no}-instance.

	\end{itemize}
	\medskip

	\noindent{\em Case 2}:  $G^{\star}$ is obtained from $G$ after contacting the edges of a matching in $G.$

	\begin{itemize}
		\item[(a)] If the recursive call on $G^\star$ reports that $\tw(G^{\star})\leq c,$  then we do the following.
		      We first notice that the fact that $\tw(G^{\star})\leq c$ implies that $\tw(G)\leq 2c,$
		      since we can obtain a tree decomposition $({\cal T},\chi)$ of $G$ from a tree decomposition $({\cal T}^{\star},\chi^{\star})$ of $G^{\star},$
		      by replacing, in every $t\in{\cal T}^{\star},$ every occurrence of a vertex of $G^{\star}$ that is a result of an edge contraction by its endpoints in $G.$
		      Thus, we can call the algorithm of \autoref{label_confiscating} with input $(G, 2c,k)$ (which runs in $2^{{\cal O}_{s_{\cal F}}(c\log c)}\cdot n$-time) in order to find, if it exists, a set $S$ such that $|S|\leq k$ and ${\cal F}\nleq_{{\sf m}} G\setminus S.$ We distinguish again two cases.

		      \begin{itemize}
			      \item[$\bullet$] If such a set $S$ does not exist, then the algorithm reports that $(G,k)$ is a \no-instance.
			      \item[$\bullet$] If such a set $S$ exists, then we apply the algorithm of \autoref{label_panathinaiko} with input $(G\setminus S,2c)$ (which runs in $2^{{\cal O}(c)}\cdot n$-time) and we get a tree decomposition of $G\setminus S$ of width at most $10c+4.$
			            Using this decomposition, we call the algorithm of \autoref{label_sleepwalkers} for $G\setminus S$ in order to check whether it contains an elementary $r$-wall $W$ as a minor.
			            This algorithm runs in $2^{{\cal O}(c \cdot \log c)}\cdot r^{{\cal O}(c)}\cdot 2^{{\cal O}(r^{2})}\cdot n=2^{{\cal O}_{s_{\cal F}}((r+k) \cdot \log (r+k))}\cdot r^{{\cal O}_{s_{\cal F}}(r+k)}\cdot 2^{{\cal O}(r^2)}\cdot n=2^{{\cal O}_{s_{\cal F}}(r^2+(r+k) \cdot \log (r+k))}\cdot n$-time, since $|E(G\setminus S)|={\cal O}(n)$ and $|E(W)|={\cal O}(r^2).$
			            If this algorithm outputs an elementary $r$-wall $W$ of $G\setminus S,$ then we  output $W.$ Otherwise,we can safely report, because of   \autoref{label_bureaucracies}, that $\tw(G)\leq \funref{label_accouchement}(s_{\cal F})\cdot r+k=c.$
		      \end{itemize}

		\item[(b)] 	If the recursive call on $G^\star$ outputs an $r$-wall $W^\star$ of $G^{\star},$  then by
		      uncontracting the edges of $M$ in $W^\star$ we can return an $r$-wall of $G.$

		\item[(c)] 	If $(G^{\star},k)$ is a \no-instance, then we report that $(G,k)$ is also a \no-instance.
	\end{itemize}
	It is easy to see that the running time of the above algorithm is $$T(n,k,r)\ \leq\  T\left((1-\frac{1}{12c^{2}})\cdot n,k,r\right)+ 2^{{\cal O}_{s_{\cal F}}(r^2+(r+k)\log (r+k))}\cdot n,$$
	which implies that $T(n,k,r)=2^{{\cal O}_{s_{\cal F}}(r^2+(r+k)\log (r+k))}\cdot n,$ as claimed.
\end{proof}

\subsection{The algorithm}\label{label_desvanecerse}

In this subsection we prove that, in the case where $a_{\cal F}=1,$
there is an algorithm that solves  \mnb{\sc  ${\cal F}$-M-Deletion} in time $2^{{\cal O}_{s_{\cal F}}(k^{2(c+2)})} \cdot n^2,$ where
$c=c_{a, \ell_{\cal F}}$ and $a=\funref{label_schematization}(s_{\cal F}).$ Note that the existence of such an algorithm implies  \autoref{label_entrepreneur}.

Let $G$ be graph and let $W$ be a wall of $G.$
The {\em drop}, denoted by $D_{W'},$  of a subwall $W'$ of $W$
is defined as follows.
If contract in $G$ the perimeter of $W$ to a single vertex $v,$ $D_{W'}$ is the
unique 2-connected component of the resulting graph that contains the interior of $W'.$ We call the vertex~$v$
the {\em pole} of the drop $D_{W'}.$

Our algorithm avoids iterative compression
in a similar fashion as done by Marx and Schlotter in~\cite{MarxS07obta} for the \mnb{\sc Vertex Planarization} problem. 	 The algorithm has  three main steps.
We first  set $a=\funref{label_schematization} (s_{\cal F})$  and we define $d= \funref{label_preoccupations}(a, \ell_{\cal F}),$
\begin{align*}
	\beta =   & \ \funref{label_objectivement}(0,\ell_{\cal F},3,k),                                              &
	\lambda = & \ \funref{label_disbnguished}(s_{\cal F})  \cdot (a+1),                                            &
	q =       & \ \funref{label_accroissement}(s_{\cal F}),                                                            \\
	\eta =    & \ \funref{label_internalization}(\lambda+1,\beta,q),                                                  &
	z =       & \ {{\sf odd}}(\max\{\funref{label_biancheggiava}(s_{\cal F}),\eta\}),                               &
	w =       & \ \funref{label_distinguimos}(z,a,0,d),                                                               \\
	b =       & \  3+ \funref{label_inconsummabile}(s_{\cal F})\cdot w,                                               &
	l=        & \ \funref{label_unterschieden}(1,s_{\cal F},k) \cdot (k+a),                                          &
	p=        & \ \funref{label_straightforward}(1,s_{\cal F},k),                                                        \\
	h =       & \ \funref{label_internalization}(l+1, b,p),                                                           &
	r =       & \ {{\sf odd}}(\max\{\funref{label_guildenstern}(1,s_{\cal F},k),h\}), \text{ and}                  &
	R =       & \ {{\sf odd}}(\funref{label_discernimiento}(s_{\cal F})\cdot r +k)={\cal O}_{s_{\cal F}}(k^{c+2}).
\end{align*}
\paragraph{Step 1.} Run the algorithm of \autoref{label_improvements} with input $(G,R,k)$
and, in $2^{{\cal O}_{s_{\cal F}}(k^{2(c+2)})}\cdot n$-time,  either report a \no-answer, or conclude  that $\tw(G)\leq \funref{label_accouchement}(s_{\cal F})\cdot R+k$ and solve  \mnb{\sc  ${\cal F}$-M-Deletion} in $2^{{\cal O}_{s_{\cal F}}(k^{c+2} \cdot \log k)} \cdot n$-time using the algorithm of \autoref{label_confiscating}, or obtain an $R$-wall $\tilde{W}$ of  $G.$ In the third case, consider all the {${R\choose b}^2=2^{{\cal O}_{s_{\cal F}}(k^c\log k)}$ $b$-subwalls of $\tilde{W}$
and for each one of them, say $W,$ construct its drop $D_{W},$ and run the algorithm  of \autoref{label_dishonorable} with input
$(D_{W} \setminus \{v_W\}, w,s_{\cal F}),$ where $v_W$ is the pole of $D_W.$
This takes time $2^{{\cal O}_{s_{\cal F}}(k^{c})} \cdot n.$ If for some of these drops the result is
a set $A\subseteq V(D_W\setminus \{v_W\})$  with $|A|\leq a$ and a regular flatness pair $(W',\frak{R}')$ of $(D_W\setminus   \{v_W\})\setminus A$ of height $w,$ then proceed to Step~2, otherwise proceed to Step~3.

\paragraph{Step 2.}
We apply the algorithm of \autoref{label_highlighting} with input $(z,a,0,d, D_{W}\setminus \{v_W\}, A, W',\frR')$, which outputs a  flatness pair $(\breve{W}, \breve{\frR})$ of $(D_{W}\setminus \{v_W\}) \setminus A$ of height $z$ that is $d$-homogeneous with respect to ${A \choose \leq 0} =\{\emptyset\}$ and is a $W^*$-tilt of $(W',\frR')$ for some subwall $W^*$ of $W'.$
This takes $2^{{\cal O}_{s_{\cal F}}(k\log k)}\cdot n$-time.
By \autoref{label_ressemblances}, $(\breve{W}, \breve{\frR})$ is regular.
Let $A^\star:=A\cup \{v_W\}$ and keep in mind that  $(D_W\setminus   \{v_W\})\setminus A = D_W \setminus A^\star.$
Consider all the $\beta$-subwalls of $\breve{W}$, {which are at most ${z \choose \beta}^2$ many},
and for each of them, say $\hat{W},$ check in linear time whether there  is an edge, in $D_W,$ between $A^\star$ and $V(\cupall{\sf influence}_{\breve{\frR}}(\hat{W})).$
If this is the case for every such a subwall,  then proceed to Step~3. If not, let $\hat{W}$ be a $\beta$-subwall of $\breve{W}$ such that  no vertex of $A^\star$ is adjacent to $V(\cupall{\sf influence}_{\breve{\frR}}(\hat{W})).$
By applying \autoref{label_protectively} for the flatness pair $(\breve{W}, \breve{\frR})$ of $D_{W}\setminus A^\star$ and the subwall $\hat{W}$ of $\breve{W},$ we obtain in linear time a $\hat{W}$-tilt $(W'',\frR'')$ of $(\breve{W}, \breve{\frR}).$
Keep in mind that $(W'',\frR'')$ is a flatness pair of $D_W \setminus A^\star$ which is also regular and $d$-homogeneous with respect to $\{\emptyset\},$ due to \autoref{label_ressemblances} and \autoref{label_superintendent}, respectively.
Also, notice that, since $(W'',\frR'')$ is a $\hat{W}$-tilt  of $(W',{\frR}'),$ ${\sf compass}_{\frR''}(W'')$ is a subgraph of  $\cupall{\sf influence}_{{\frak{R}}'}(\hat{W})$ and therefore no vertex of $A^\star$ is adjacent, in $D_W,$ to a vertex of ${\sf compass}_{\frR''}(W'').$
{The latter implies that we can obtain a $7$-tuple $\tilde{\frR}''$ from $\frR''$ by adding all vertices of $G\setminus V({\sf compass}_{\frR''}(W''))$ such that $ {\sf compass}_{\tilde{\frR}''}(W'')=  {\sf compass}_{\frR''}(W'')$ and $(W'',\tilde{\frR}'')$ is a flatness pair of $G.$
We apply  {\tt Find-Irrelevant-Vertex} of \autoref{label_mitinbegriffen} with input $(k,0,G,{\emptyset}, W'',\tilde{\frR}'')$
and obtain, in  linear time,  an  irrelevant vertex $v$
such that $(G,k)$ and $(G\setminus v,k)$
are equivalent instances of \mnb{\sc  ${\cal F}$-M-Deletion}.
Then the algorithm runs recursively on the equivalent instance $(G\setminus v,k).$ \smallskip

\noindent Notice that Step~2 can be seen as the
	{\sl irrelevant vertex case} of our algorithm.

\paragraph{Step 3.}
Consider all  the $r$-subwalls of $\tilde{W},$ {which are at most ${R\choose r}^2=2^{{\cal O}_{s_{\cal F}}(k^c\log k)}$ many,}
and for each of them,
compute its canonical partition ${\cal Q}$.
Then, for each $p$-internal bag $Q$ of ${\cal Q},$ add a new vertex $v_Q$ and make it
adjacent to all vertices in $Q,$ then
add a new vertex $x_{\rm all}$ and make it adjacent  to  all $x_{Q}$'s, and in the resulting graph, for every vertex $y$ of $G$ that is not in the union of the internal bags of ${\cal Q},$ check, in time ${\cal O}(k\cdot |E(G)|)={\cal O}_{s_{\cal F}}(k\cdot n)$ (using standard flow techniques), whether there are $\funref{label_unterschieden}(1,s_{\cal F},k)$
internally vertex-disjoint paths from $x_{\rm all}$ to $y.$  If this is indeed the case for some $y,$ then  $y$ should belong to every solution of \mnb{\sc  ${\cal F}$-M-Deletion} for the instance $(G,k),$
and the algorithm runs recursively  on the equivalent  instance $(G\setminus y,k-1).$
If no such a vertex $y$ exists, then report that $(G,k)$ is a \no-instance of \mnb{\sc  ${\cal F}$-M-Deletion}. \smallskip

\noindent Note that Step 3 can be seen as a trivial
	{\sl branching case} where the only   choice is vertex $y.$

\medskip

Notice that the third step of the algorithm,
when applied takes time $2^{{\cal O}_{s_{\cal F}}(k^c\log k)} \cdot n^2.$
However, it cannot be applied more than $k$ times during the course of the algorithm. As the first step runs in time
$2^{{\cal O}_{s_{\cal F}}(k^{2(c+2)}\log k)} \cdot n,$ and the second step runs in time $2^{{\cal O}_{s_{\cal F}}(k\log k)}\cdot n,$
they may be applied at most $n$ times, and the claimed time complexity follows.\\

\subsection{Correctness of the algorithm}\label{label_sleeplessness}
In this subsection we prove the correctness of the algorithm presented in  \autoref{label_desvanecerse}.

Recall that $a=\funref{label_schematization} (s_{\cal F}),$ $d= \funref{label_preoccupations}(a, \ell_{\cal F}),$
\begin{align*}
	\beta =   & \ \funref{label_objectivement}(0,\ell_{\cal F},3,k),                                              &
	\lambda = & \ \funref{label_disbnguished}(s_{\cal F})  \cdot (a+1),                                            &
	q =       & \ \funref{label_accroissement}(s_{\cal F}),                                                            \\
	\eta =    & \ \funref{label_internalization}(\lambda+1,\beta,q),                                                  &
	z =       & \ {{\sf odd}}(\max\{\funref{label_biancheggiava}(s_{\cal F}),\eta\}),                               &
	w =       & \ \funref{label_distinguimos}(z,a,0,d),                                                               \\
	b =       & \  3+ \funref{label_inconsummabile}(s_{\cal F})\cdot w,                                               &
	l=        & \ \funref{label_unterschieden}(1,s_{\cal F},k) \cdot (k+a),                                          &
	p=        & \ \funref{label_straightforward}(1,s_{\cal F},k),                                                        \\
	h =       & \ \funref{label_internalization}(l+1, b,p),                                                           &
	r =       & \ {{\sf odd}}(\max\{\funref{label_guildenstern}(1,s_{\cal F},k),h\}), \text{ and}                  &
	R =       & \ {{\sf odd}}(\funref{label_discernimiento}(s_{\cal F})\cdot r +k)={\cal O}_{s_{\cal F}}(k^{c+2}).
\end{align*}

Let  $(G,k)$ be a \yes-instance and let $S$ be a solution, i.e., a subset of $V(G)$ of size at most $k$ such that $G\setminus S\in{\bf exc}({\cal F})$ and let $\tilde{W}$ be an $R$-wall of  $G.$
Then, since $R\geq \funref{label_discernimiento}(s_{\cal F})\cdot r+k,$ there is an $(\funref{label_discernimiento}(s_{\cal F})\cdot r)$-subwall
of $\tilde{W},$ say $W^{\star},$ that does not contain vertices of $S.$
The wall $W^{\star}$ is an $(\funref{label_discernimiento}(s_{\cal F})\cdot r)$-wall
of $G\setminus S$, and therefore by \autoref{label_propositional} there is a set
$A\subseteq V(G\setminus S),$ with $|A|\leq a,$
and a flatness pair $(W,\frR)$ of $G\setminus (S\cup A)$ of height $r.$

Let $\tilde{\cal Q}$ be a $(W,\frR)$-canonical partition of $G\setminus (S\cup A).$
For each $p$-internal bag $Q$ of $\tilde{\cal Q},$ add a new vertex $v_Q$ and make it
adjacent to all vertices in $Q,$ then
add a new vertex $x_{\rm all}$ and make it adjacent  to  all $v_{Q}$'s.
In the resulting graph, if there are $\funref{label_unterschieden}(1,s_{\cal F},k)$ internally vertex-disjoint paths from $x_{\sf all}$ to a vertex $v\in S\cup A,$ then this is checked in  Step 3 (since connectivity of the internal bags implies that every such a path can be rerouted in order to intersect the wall) and the algorithm correctly (due to \autoref{label_disviluppato}) runs recursively on the equivalent instance $(G\setminus v,k-1).$
If this is not the case, then for each vertex $v$ of $S\cup A$ there are less than $\funref{label_unterschieden}(1,s_{\cal F},k)$ $p$-internal bags of $\tilde{\cal Q}$ that contain vertices adjacent to $v.$
This means that the $p$-internal bags
of $\tilde{\cal Q}$ that contain vertices adjacent to some vertex of $S\cup A$ are at most $\funref{label_unterschieden}(1,s_{\cal F},k)\cdot (k+a) =l.$

%

We consider a family ${\cal W}=\{{W}_{1}, \ldots, {W}_{l+1}\}$ of $l+1$ $b$-subwalls of $W$ such that for every $i \in [z],$ $\cupall{\sf influence}_{\frR}(W_i)$ is a subgraph of $\cupall \{Q\mid Q \text{ is a $p$-internal bag of }\tilde{\cal Q}\}$ and for every $i,j\in[z],$ with $i\neq j,$ there is no internal bag of $\tilde{\cal Q}$ that contains vertices of both $V(\cupall{\sf influence}_\frR (W_i))$ and $V(\cupall{\sf influence}_\frR (W_j)).$ The existence of ${\cal W}$ follows from the fact that $r\geq h = \funref{label_internalization}(l+1,b,p)$ and \autoref{label_stereotypical}.

The fact that the $p$-internal bags
of $\tilde{\cal Q}$ that contain vertices adjacent to some vertex of $S\cup A$ are at most $l$ implies that
there exists an $i\in[l+1]$ such that
no vertex of $V(\cupall{\sf influence}_{\hat{\frR}}({W_i}))$ is adjacent, in $G,$ to a vertex in $S\cup A.$
Thus, if $D_{W_i}$ is the drop of $W_{i},$ and $D_{W_i}^-:=D_{W_i}\setminus \{v_{W_i}\},$ where $v_{W_{i}}$ is the pole of $D_{W_i},$ then $D_{W_i}^-\leq_{\sf m} G\setminus (S\cup A).$
This, in turn, implies that $D_{W_i}^-\in {\bf exc}({\cal F})$ and therefore $D_{W_i}^-$ does not contain $K_{s_{\cal F}}$ as a minor.
Additionally, we notice that $D_{W_i}^-$ contains the central $(b-2)$-subwall $\bar{W_{i}}$ of $W_i$ as a subgraph and since $\bar{W_i}$ has height $b-2= \funref{label_inconsummabile}(s_{\cal F})\cdot w + 1,$ it holds that $\tw(D_{W_i}^-)> \funref{label_inconsummabile}(s_{\cal F})\cdot w.$
Therefore, by applying the algorithm  of \autoref{label_dishonorable} with input
$(D_{W_i}^- , w,s_{\cal F}),$ we must find a set $A\subseteq V(D_{W_i}\setminus \{v_{W_i}\})$  with $|A|\leq a$ and a regular flatness pair $(W',\frak{R}')$ of $D_{W_i}^- \setminus A$ of height $w.$
This should be detected in Step~1.

We apply the algorithm of \autoref{label_highlighting} with input $(z,a,0,d, D_{W_i}^-, A, W',\frR')$, which outputs a flatness pair $(\breve{W}, \breve{\frR})$ of $D_{W_i}^- \setminus A$ of height $z$ that is $d$-homogeneous with respect to ${A \choose \leq 0} = \{\emptyset\}$ and is a $W^*$-tilt of $(W',\frR')$ for some subwall $W^*$ of $W'.$
We set $A^\star:=A\cup \{v_{W_i}\}$ and keep in mind that $D_{W_i}\setminus A^\star = D_{W_i}^-\setminus A.$
Let $\tilde{\cal Q}'$ be a $(\breve{W},\breve{\frR})$-canonical partition of $D_{W_i}\setminus A^\star.$
Since $D_{W_i}$ is a subgraph of $D_{W_i}^-$ and $D_{W_i}^-\in {\bf exc}({\cal F}),$  we observe that $D_{W_i}\in {\bf exc}({\cal F}).$ Hence, as a consequence of \autoref{label_persuadieran},
every vertex in $A^\star$ has neighbors in less than $\funref{label_disbnguished}(s_{\cal F})$ $q$-internal bags of ${\cal Q}''.$
Therefore, since $\lambda= \funref{label_disbnguished}(s_{\cal F})\cdot (a+1)\geq \funref{label_disbnguished}(s_{\cal F}) \cdot |A^\star|$
and
$\rho\geq \eta =\funref{label_internalization}(\lambda+1,\beta,q),$
it follows, due to \autoref{label_stereotypical}, that
there is a $\beta$-subwall $\bar{W}$ of $W'$
such that no vertex of $V(\cupall{\sf influence}_{\frR'}(\bar{W}))$ is adjacent, in $D_{W_i},$
to a vertex in $A^\star.$
Therefore, this $\beta$-subwall $\bar{W}$ of $W'$ should be detected in Step~2.

\section{Algorithms for variants of {\sc  Vertex Deletion to ${\cal G}$}}\label{label_ressentiment}

We now present how our approach can be modified so to obtain \FPT-algorithms for several variants of the \mnb{\sc  Vertex Deletion to ${\cal G}$} problem.

\subsection{The general framework}
Notice that both algorithms in \autoref{label_pizzighettona} and \autoref{label_demonstrieren} are based on one of the following three  scenarios
for \mnb{\sc  Vertex Deletion to ${\cal G}$} with input $(G,k).$

\begin{itemize}
	\item[] [{\sl Bounded treewidth case}] A tree decomposition of $G$ of width $k^{{\cal O}(1)},$ or
	\item[] [{\sl Branching case}] a set $B$ with $|B|={\cal O}(k)$ such that $(G,k)$ is a \yes-instance if and only if, for some $x\in B,$  $(G\setminus x,k-1)$ is a \yes-instance, or
	\item[] [{\sl Irrelevant vertex case}] a vertex $x$ such that $(G,k)$ is a \yes-instance if and only if $(G\setminus x,k)$ is a \yes-instance,
\end{itemize}
For each of the variants of  \mnb{\sc  Vertex Deletion to ${\cal G}$} that we consider, the algorithm recursively runs on an equivalent instance  with one vertex less ({\sl irrelevant vertex case}), or branches on ${\cal O}(k)$ equivalent instances   where both $k$ and the size of the graph are  one less ({\sl branching case}). The eventual outcome  is to reduce the problem to the {\sl bounded treewidth case}, producing a tree decomposition of $G$ of width $k^{{\cal O}(1)}$ ({\sl bounded treewidth case}).
In each variant of the problem,  the  {\sl bounded treewidth case} can be treated by a suitable modification of the dynamic programming algorithm of \cite{BasteST20hittI}, taking into account the main combinatorial result in~\cite{BasteST20acom}. For each variant that we treat, the algorithm of \autoref{label_pizzighettona}  assumes that we have at hand a solution of \mnb{\sc  Vertex Deletion to ${\cal G}$} of size $k,$ which can be found by the algorithm in \autoref{label_preliminaries}.\medskip

We next present the problem variants and explain how to adapt the {\sl branching case}  and the {\sl irrelevant vertex case} for each of them.

\subsection{Variants of {\sc  Vertex Deletion to ${\cal G}$}}
\label{label_enlightening}

A common part of the inputs of all problems below is the pair $(G,k),$
where $G$ is a graph and $k$ is a non-negative integer, i.e., an input of \mnb{\sc  Vertex Deletion to ${\cal G}$}.

\paragraph{Annotated.}
In the {\em  annotated version} of \mnb{\sc  Vertex Deletion to ${\cal G}$},
the input is a triple $(G,k,R),$ where  $R\subseteq V(G),$
and the problem asks for a solution $S$ with $S\subseteq R.$

\begin{itemize}
	\item[] [{\sl Branching case}]: we branch on $(G\setminus x,k-1,R\setminus x)$  for all the annotated vertices
	      of $B,$ i.e., for every $x\in B\cap R.$ If there is no such a vertex, we report that $(G,k,R)$
	      is a \no-instance.
	\item[] [{\sl Irrelevant vertex case}]: we recurse on $(G\setminus x,k,R\setminus x),$ as every irrelevant vertex for the original  problem is also an irrelevant vertex for its annotated variant.
\end{itemize}

\paragraph{Modulo.}
In the {\em modulo version} of \mnb{\sc  Vertex Deletion to ${\cal G}$},  the input is a quadruple $(G,k,q,p)$ where  $q,p$ are integers, $p$ is a prime, and $q<p.$
The question is whether there is a solution $S$ of size at most $k$
where $|S| \equiv q \pmod p.$

\begin{itemize}
	\item[] [{\sl Branching case}]: we branch on $(G\setminus x,k-1,q-1 \pmod p,p)$  for every $x\in B.$
	\item[] [{\sl Irrelevant vertex case}]: it is the same, as every irrelevant vertex for the original  problem is also an irrelevant vertex for this variant.
\end{itemize}

\paragraph{Weighted.}
In the {\em weighted version} of \mnb{\sc  Vertex Deletion to ${\cal G}$},  the input is a triple $(G,k,\textbf{w})$ where
$\textbf{w}: V(G)\to\Bbb{R}$ is a weight function assigning positive real weights to the vertices o $G.$ The problem asks for a solution $S$ with $\sum_{v\in S}\textbf{w}(v)\leq k.$

\begin{itemize}
	\item[] [{\sl Branching case}]: we branch on $(G\setminus x,k-\textbf{w}(x), \textbf{w}\setminus\{(x,\textbf{w}(x))\}),$ for every $x\in B.$
	\item[] [{\sl Irrelevant vertex case}]: it is the same, as every irrelevant vertex for the original  problem is also an irrelevant vertex for its weighted variant.
\end{itemize}
For the above problem, if $\varepsilon=\min\{\textbf{w}(x)\mid x\in V(G)\},$ then  the parametric dependence of the derived algorithm  is  $2^{{\sf poly}(k/\varepsilon)},$
as the size of the solution $S$ is at most $k/\varepsilon.$

\paragraph{Counting.}
In the  {\em counting version} of \mnb{\sc  Vertex Deletion to ${\cal G}$}   with input $(G,k),$
the output is the number $\#_{\cal G}(G,k)$ of all solutions  of \mnb{\sc  Vertex Deletion to ${\cal G}$}
of size (at most) $k.$  We treat the case where we count solutions of size exactly $k$ as the ``$\leq k$''-case can be easily reduced to it.

\begin{itemize}
	\item[] [{\sl Branching case}]: return $\sum_{x\in B}\#_{\cal G}(G\setminus x,k-1).$

	\item[] [{\sl Irrelevant vertex case}]: return $\#_{\cal G}(G\setminus x,k-1)+\#(G\setminus x,k).$
\end{itemize}

The above creates $T(n,k)$
subproblems on bounded treewidth graphs, where $$T(n,k)=\max\{{\cal O}(k)\cdot T(n-1,k-1), T(n-1,k-1)+T(n-1,k)\}.$$ This makes a total of $2^{{\cal O}(k)}\cdot n$ problems,
each solvable in $2^{k^{{\cal O}(1)}}\cdot n$-time by the counting version of the  dynamic programming algorithm of~\cite{BasteST20hittI}, taking into account the analysis of \cite{BasteST20acom}.


\paragraph{Colored.}
In the  {\em colored version} of \mnb{\sc  Vertex Deletion to ${\cal G}$}, the input is
a triple $(G,k,\chi)$ where  $\chi: V(G)\to [k]$ is a function assigning colors from $[k]$ to the vertices of $G.$
The problem asks for a solution $S$ to \mnb{\sc  Vertex Deletion to ${\cal G}$}
that carries all $k$ colors, i.e.,  for each $i\in [k],$ $|S\cap \chi^{-1}(i)|=1.$
(Notice that the requested solution must have size exactly $k.$)
To deal with this problem, we deal with its annotated version where
we permit $\chi: V(G)\to \{0,1,\ldots,k\},$ i.e., the vertices  in $R:=\bigcup_{i\in[k]}\chi^{-1}(i)$
are annotated, while the vertices in $\chi^{-1}(0)$ cannot participate  in a solution (we call these vertices {\em black} vertices).

\begin{itemize}
	\item[] [{\sl Branching case}]: we branch on $(G\setminus x,k-1,\chi|_x),$ for every $x\in B\cap R,$ where $$\chi|_{x}=\{(v,\chi(v))\mid v\in V(G)\setminus \chi^{-1}(\chi(x))\}\cup\{(v,0)\mid v\in \chi^{-1}(\chi(x))\setminus\{x\}\}.$$
	      The new coloring $\chi|_{x}$ turns black all vertices carrying  the color of $x.$
	      If $ B\cap R=\emptyset,$ i.e., all vertices in $B$ are black, then we have a \no-instance.

	\item[] [{\sl Irrelevant vertex case}]:  Before we apply the irrelevant vertex case, we check whether there is some $i\in  [k]$ with $
		      |\chi^{-1}(i)|\leq 1,$ i.e., there is a color in $[k]$ that appears once or is not used at all. If $|\chi^{-1}(i)|=0,$ then we return that we have a \no-instance. If $\chi^{-1}(i)=\{x\},$ then $x$ should belong to every possible  solution and, in this case, we recurse on $(G\setminus x,k-1,\chi\setminus\{(x,\chi(x))\}).$
	      If now each color is used at least twice, we recurse on $(G\setminus x,k,\chi\setminus\{(x,\chi(x))\}),$ i.e., apply the irrelevant vertex case.

\end{itemize}

\section{Discussion and concluding remarks}\label{label_grundgesetzen}

\paragraph{Apices of topological minors.}
Very recently, Fomin et al.~\cite{FominLPSZ20hitt} gave an \FPT-algorithm running in time  ${\cal O}_{s,k}(n^4)$ for the following  problem: for a fixed finite family of graphs ${\cal F},$ each on at most $s$ vertices, decide whether an $n$-vertex input graph $G$ is a $k$-apex of the class of graphs that exclude the graphs in ${\cal F}$ as \emph{topological minors}\footnote{The definition is as minors, except that only edges incident with at least a degree-two vertex can be contracted.}. For every  graph $H,$ there is a finite set ${\cal H}$ of graphs such that a graph $G$ contains $H$ as a minor if and only if $G$ contains a graph in ${\cal H}$ as a topological minor. Based on this observation, the result of Fomin et al.~\cite{FominLPSZ20hitt}  implies
that for every minor-closed graph class ${\cal G},$
\mnb{\sc Vertex Deletion to ${\cal G}$} admits  an ${\cal O}(h(k,s)\cdot n^4)$-time \FPT-algorithm,
where $s$ is the maximum size of an obstruction of ${\cal G}.$ Notice that this implication is a solid improvement on \mnb{\sc  Vertex Deletion to ${\cal G}$}
with respect to the result of \cite{AdlerGK08comp},
where only the computability of $h$ is proved.
However, as mentioned in~\cite{FominLPSZ20hitt}, even for fixed values of $s,$
the dependence of $h$ on $k$ is humongous.
Therefore, \autoref{label_preliminaries}  can be seen as  orthogonal to the result  of~\cite{FominLPSZ20hitt}. An interesting question is whether the ideas of this paper can be useful towards improving the the parametric dependence of the algorithm of~\cite{FominLPSZ20hitt}.

\paragraph{Limitations of the irrelevant vertex technique.} An intriguing open question is whether
\mnb{\sc  Vertex Deletion to ${\cal G}$} admits a $2^{{\cal O}_{s_{\cal F}}\left(k^{c}\right)} \cdot n^{{\cal O}(1)}$-time algorithm
for some universal
constant $c$ that does {\sl not} depend on the class ${\cal G}$. Clearly, this is not the
case of the algorithms of  \autoref{label_preliminaries} and \autoref{label_entrepreneur}, running in time $2^{{\cal O}_{s_{\cal F}}(k^{2(c+2)})}\cdot n^3$ and $2^{{\cal O}_{s}(k^{2(c+2)})}\cdot n^2,$ respectively, where
$c$ is the {\sl palette-variety} of the minor-obstruction set ${\cal F}$ of  ${\cal G}$ which, from the corresponding proofs,  is
estimated to be $c=2^{2^{{\cal O}(s^2\cdot \log s)}}$ and $c=2^{2^{{\cal O}(s^{24}\cdot \log s)}},$   respectively (recall that $s$ is the maximum size of a minor-obstruction of ${\cal G}$).
We tend to believe that this dependence is unavoidable if we want to use the  {\sl irrelevant vertex technique},
as it reflects the {\em price of homogeneity}, mentioned in~\autoref{label_definitionen}.
Having homogeneous walls is critical for the
application of this technique (see \autoref{label_appressavamo}) when ${\cal G}$
is more general  than surface-embeddable graphs (in the bounded genus case, all subwalls are already homogeneous).  Is there a way to prove that this behavior is unavoidable
subject to some complexity assumption?
{An interesting result of this flavor, concerning the existence of polynomial kernels for \mnb{\sc  Vertex Deletion to ${\cal G}$}, was given by Giannopoulou et al.~\cite{GiannopoulouJLS17unif} who  proved that, even for minor-closed families ${\cal G}$ that exclude a planar graph, the dependence on ${\cal G}$ of the degree of the polynomial kernel,  which exists because of~\cite{FominLMS12plan}, is unavoidable subject to reasonable complexity assumptions.}

\paragraph{Kernelization.}
As mentioned above, Giannopoulou et al.~\cite{GiannopoulouJLS17unif} provided a polynomial kernel for \mnb{\sc  Vertex Deletion to ${\cal G}$} in the case where ${\cal G}$ excludes a planar graph.
To the best of our knowledge, the existence of a polynomial kernel is open for every family ${\cal G}$ whose obstructions are all non-planar.
In particular, no polynomial kernel is known even for the \mnb{\sc Vertex Planarization} problem.
%

\paragraph{Other modification operations.}
Another direction is to consider graph modification to a minor-closed
graph class for different modification operations. Our approach becomes just simpler in the case  where the modification operation is edge removal or edge contraction. In these two cases, {we immediately get rid of the branching part of our algorithms}  and only the {irrelevant vertex}  part needs to be applied. Another challenge  is to combine all aforementioned modifications. This  is more complicated (and tedious) but not more complex. What is really more complex is to consider as well edge additions. We leave it as an open research challenge (a first step was done for the case of planar graphs~\cite{FominGT19modi}).

\paragraph{Lower bounds.}
Concerning lower bounds for \mnb{\sc  Vertex Deletion to ${\cal G}$} under the Exponential Time Hypothesis~\cite{ImpagliazzoP01whic}, we are not aware of any lower bound stronger than $2^{o(k)} \cdot n^{\bigO{1}}$ for {\sl any} minor-closed class ${\cal G}.$
This lower bound already applies when  ${\cal F}=\{K_{2}\},$
i.e., for the \mnb{\sc Vertex Cover} problem~\cite{ImpagliazzoP01whic,BasteST20hittIII}.

\medskip

\paragraph{Acknowledgements.}
We wish to thank the anonymous reviewers of the conference version of this paper for their comments and remarks that improved the presentation
of this manuscript.
Moreover, we wish to thank Dániel Marx
for his valuable remarks and advises  {regarding the variants of the problem discussed in \autoref{label_ressentiment}}.

%
%

\end{document}